\documentclass[reqno]{amsart}
\usepackage{amsfonts}
\usepackage{mathrsfs}
\usepackage{color}
\usepackage{cite}
\usepackage{amscd}
\usepackage{latexsym}
\usepackage{amsfonts}
\usepackage{graphicx}
\usepackage{CJK,indentfirst,amsmath,amsfonts,amssymb,amsthm,cite,cases,subeqnarray,setspace}

\begin{document}

\title[Quasineutral limit of the Euler-Poisson equation]
{Quasineutral limit of the Euler-Poisson equation for a cold, ion-acoustic plasma}
\author[X. Pu and B. Guo]
{Xueke Pu, and Boling Guo}

\address{Xueke Pu \newline
Department of Mathematics, Chongqing University, Chongqing 400044, P.R. China; Mathematical Sciences Research Institute in Chongqing.}
\email{ xuekepu@cqu.edu.cn}

\address{Boling Guo \newline
Institute of Applied Physics and Computational Mathematics, P.O. Box 8009, Beijing, China, 100088.}
\email{gbl@iapcm.ac.cn}

\thanks{This work is supported by NSFC (grant 11001285).}

\begin{abstract}
In this paper, we consider the quasineutral limit of the Euler-Poisson equation for a clod, ion-acoustic plasma when the Debye length tends to zero. When the ion-acoustic plasma is cold, the Euler-Poisson equation is pressureless and hence fails to be Friedrich symmetrisable, which excludes the application of the classical energy estimates method. This brings new difficulties in proving uniform estimates independent of $\varepsilon$. The main novelty in this article is to introduce new $\varepsilon$-weighted norms of the unknowns and to combine energy estimates in different levels with weights depending on $\varepsilon$. Finally, that the quasineutral regimes are the incompressible Euler equations is proven for well prepared initial data.
\end{abstract}

\maketitle \numberwithin{equation}{section}
\newtheorem{proposition}{Proposition}[section]
\newtheorem{theorem}{Theorem}[section]
\newtheorem{lemma}[theorem]{Lemma}
\newtheorem{remark}[theorem]{Remark}
\newtheorem{hypothesis}[theorem]{Hypothesis}
\newtheorem{definition}{Definition}[section]
\newtheorem{corollary}{Corollary}[section]
\newtheorem{assumption}{Assumption}[section]

\section{Introduction}
In this paper, we consider the Euler-Poisson equation for a clod, ion-acoustic plasma
\begin{subequations}\label{EP}
\begin{numcases}{(EP)\ \ }
\partial_tn+div(n\bf{u})=0,\label{EP-n}\\
\partial_t\bf{u}+\bf{u}\cdot\nabla \bf{u}=-\nabla\phi,\label{EP-u}\\
\varepsilon\Delta\phi=e^{\phi}-n,\label{EP-p}
\end{numcases}
\end{subequations}
where $n$ is the density of the ions, ${\bf{u}}=(u_1,\cdots,u_d)$ is the velocity field and $\phi$ is the electric potential. Here $e^{\phi}$ is the rescaled electron density by the famous Boltzmann relation and $\varepsilon\ll1$ is a small parameter representing the squared scaled Debye length $\varepsilon={\epsilon_0k_BT_e}/{N_0e^2L^2}$, where $\epsilon_0$ is the vacuum permittivity, $L$ is the characteristic observation length and $T_e$ is the average temperature of the electrons. For typical plasma applications, the Debye length is very small compared to the characteristic length of physical interest and it is therefore necessary to consider the limiting system when $\varepsilon\to 0$. For more physical background of the Euler-Poisson equation or the ion-acoustic plasma, one may refer to \cite{WT66,KT86}.

Formally, when letting $\varepsilon\to0$, we obtain from the third equation in \eqref{EP} that $\phi=\ln n$, and hence the following compressible Euler system
\begin{equation}\label{EQ}(EQ)\ \
\begin{cases}
\partial_tn+div(n{\bf u})=0,\\
\partial_t{\bf u}+{\bf u}\cdot\nabla {\bf u}+\nabla\ln n=0.
\end{cases}
\end{equation}
This limit system \eqref{EQ} is an hyperbolic symmetrisable system, whose classical result for the existence and uniqueness of sufficiently smooth solutions in small time interval is available in \cite{Majda84}. The system \eqref{EQ} have to be supplemented by suitable initial conditions. We shall assume that the plasma is uniform and electrically neutral near infinity, i.e., $n\to n^{\pm}$ and ${\bf u}\to0$ as $x\to\pm\infty$. More precisely, let $\tilde n$ be a smooth strictly positive function, constant outside $x\in [-1,+1]$, going to $n^{\pm}$ as $x\to\infty$. We assume that the initial conditions $(n_0^0,{\bf u}_0^0)$ satisfy
\begin{equation}\label{IC}
(n^0_0-\tilde n)\in H^s({\Bbb R^d}),\ \ u_0^0\in{\bf H}^s({\Bbb R^d}),\ \ n_0^0\geq\sigma>0,
\end{equation}
for some $s>3/2$ and some constant $\sigma>0$.

\begin{theorem}\label{thm1}
Let $(n_0^{0},{\bf u}_0^{0})\in H^{s'}\times H^{s'}$ be initial data with $s'>\frac d2+1$ and satisfy \eqref{IC}. Then there exists $T>0$, maximal time of existence and a solution $(n^0,{\bf u}^0)$ of \eqref{EQ} on $0\leq t<T$ with initial data $(n_0^{0},{\bf u}_0^{0})$ such that ${\bf u}^0$ and $n^0-\tilde n$ are in $L^{\infty}([0,T'])$ for every $T'<T$. Furthermore,
$$(n^0,{\bf u}^0)\in \left(C([0,T];H^{s'})\times C([0,T];{\bf H}^{s'})\right)\bigcap \left(C^1([0,T];H^{{s'}-1})\times C^1([0,T];{\bf H}^{{s'}-1})\right)$$
and $T$ depends only on $\|(n_0^{0},{\bf u}_0^{0})\|_{H^{s'}\times {\bf H}^{s'}}$.
\end{theorem}

\begin{remark}
For the maximal existence time $T$, either $T=\infty$ in the case of global existence, or $T<\infty$ and the solution blows up when $t\to T$:
\begin{equation}\label{blowup}
\limsup_{t\to T}(\|{\bf u}^0\|_{L^{\infty}} +\|n^0\|_{L^{\infty}}+\|\frac{1}{n^0}\|_{L^{\infty}} +\int_0^t(\|\partial_x{\bf u}^0\|_{L^{\infty}} +\|\partial_xn^0\|_{L^{\infty}})d\tau)=\infty.
\end{equation}
Therefore, we will work on a time interval $[0,T']$ for $T'<T$ (but arbitrary close to $T$) in order to insure $0<\sigma'<n^0(t,x)<\sigma''$ for all $(t,x)$, for some constants $\sigma',\sigma''>0$. Here, $\sigma'$ may approach to $0$ as $T'$ goes to $T$.
\end{remark}

Let us define $\phi^0=\ln n^0$. The main result in this paper is the following

\begin{theorem}\label{thm}
Let $s'\in \Bbb{N}$ with $s'>[\frac d2]+2$ be sufficiently large. Let $(n_0^{0},{\bf u}_0^{0})\in H^{s'}\times {\bf H}^{s'}$ and $(n^0,{\bf u}^0)$ be the solution of the limit system \eqref{EQ} on $[0,T)$ with initial data $(n_0^{0},{\bf u}_0^{0})$, given in Theorem \ref{thm1}. Then there exists solutions $(n^{\varepsilon}(t),{\bf u}^{\varepsilon}(t))$ of \eqref{EP} with the same initial data on $[0,T^{\varepsilon})$ with $\liminf_{\varepsilon\to0}T^{\varepsilon}\geq T$. Moreover, for every $T'<T$ and for every $\varepsilon$ small enough, $\varepsilon^{-1}(n^{\varepsilon}-n^0)$ and $\varepsilon^{-1}({\bf u}^{\varepsilon}-{\bf u}^0)$ are bounded in $L^{\infty}([0,T'];H^{s})$ and $L^{\infty}([0,T'];{\bf H}^{s})$, respectively, for some $s<s'$.
\end{theorem}

Without essential difficulties, we can show that the same result holds on the torus $\Bbb T=\Bbb R/\Bbb Z$ or equivalently on $[0,1]$ with periodic boundary conditions following the method in the present paper.

Before proving this theorem, we make several remarks on the background of the Euler-Poisson equation and the development of its quasineutral limit. The more general isothermal Euler-Poisson equation for ion-acoustic plasma has the following form
\begin{subequations}\label{EPP}
\begin{numcases}{}
\partial_tn+div(n\bf{u})=0,\label{EPP-n}\\
\partial_t{\bf u}+{\bf u}\cdot\nabla {\bf u}+\frac{T_i}{n}\nabla n=-\nabla\phi,\label{EPP-u}\\
\varepsilon\Delta\phi=e^{\phi}-n,\label{EPP-p}
\end{numcases}
\end{subequations}
where $T_i>0$ is the ion temperature. When $T_i=0$ (comparing with the electron temperature), this equation reduces to \eqref{EP} for the cold, ion-acoustic plasma. For \eqref{EPP} (with fixed $T_i>0$), Cordier and Grenier \cite{CG00} showed the quasineutral limit as $\varepsilon\to0$ by using the pseudodifferential energy estimates method of \cite{Gre97}. It is shown that, under suitable conditions, the solution of \eqref{EPP} converges to the following Euler equation as $\varepsilon\to0$
\begin{equation*}
\begin{cases}
\partial_tn+div(n{\bf u})=0,\\
\partial_t{\bf u}+{\bf u}\cdot\nabla {\bf u}+(T_i+1)\nabla\ln n=0.
\end{cases}
\end{equation*}
However, as far as we know, there is no quasineutral limit result for \eqref{EP} so far. The main difference between \eqref{EP} and \eqref{EPP} is that \eqref{EPP} has the pressure term $T_i\nabla\ln n$, which is crucial in proving the quasineutral result of the Euler-Poisson equation \eqref{EPP}. With this term, the hyperbolic part of \eqref{EPP} is Friedrich symmetrisable and the general framework of pseudodifferential operator energy estimates methods of Grenier \cite{Gre97} can be applied. One may refer to \cite{CG00} for more details of application of this method in treating the quasineutral limit of \eqref{EPP}. But without the pressure term, as is the case in the present paper, since the hyperbolic part (the equations \eqref{EP-n} and \eqref{EP-u}) of \eqref{EP} is not symmetrisable, the pseudodifferential energy method cannot apply and no quasineutral limit can be drawn without introducing new techniques.

For the Euler-Poisson equation \eqref{EPP}, Guo and Pausader \cite{GP11} constructed global smooth irrotational solutions with small amplitude for this equation with fixed $\varepsilon>0$ and $T_i>0$. Very recently, Guo and Pu \cite{GP13} derived the KdV equation from \eqref{EPP} for the full range of $T_i\geq0$, and Pu \cite{Pu13} derived the Kadomtsev-Petviashvili II equation and the Zakharov-Kuznetsov equation via the Gardner-Morikawa type transformations. Guo \emph{et al} \cite{GIP13} made a breakthrough for the Euler-maxwell two-fluid system in 3D and proved that irrotational, smooth and localized perturbations of a constant background with small amplitude lead to global smooth solutions. We also would like to remark that Loeper \cite{Loeper05} proved quasineutral limit results recently for the electron Euler-Poisson equation without pressure term and the Euler-Monge-Amp\`{e}re equation, whose method is different from ours and cannot be applied to our situation. For numerical studies for the pressureless Euler-Poisson equation \eqref{EP}, the reader may refer to a recent paper of Degond \emph{et al} \cite{DLSV10}, which analyzes various schemes for the Euler-Poisson-Boltzmann equation. For more results on the quasi-neutral limit results of the Euler-Poisson equation and related models, one may refer to various recent papers and the references therein, see \cite{Bre00, GHR11, Gre97a, JJLL10, SS01, Wang04, Wang08, WJ06, SLT01} to list only a few.

The next section is devoted to the proof of Theorem \ref{thm}. For this purpose, we write the solution of \eqref{EP} as $n^{\varepsilon}=n^0+\varepsilon n^1$ and ${\bf u}^{\varepsilon}={\bf u}^0+\varepsilon {\bf u}^1$ and consider the remainder system $(R_{\varepsilon})$ of $n^1$ and ${\bf u}^1$. The main idea is then to show that $(n^1,{\bf u}^1)$ is uniformly bounded in $H^s\times {\bf H}^s$ when $\varepsilon\to 0$. To overcome the difficulty of non-symmetrisability of \eqref{EP}, we introduce some triple norm $|||\cdot|||_{\varepsilon,s}$
\begin{equation}\label{trinorm}
\begin{split}
|||n^1|||^2_{\varepsilon,s}=&\|n^1\|^2_{H^s},\\
|||{\bf u}^1|||^2_{\varepsilon,s}=&\|{\bf u}^1\|^2_{H^s}+\varepsilon\|\nabla{\bf u}^1\|^2_{H^s},\\
|||\phi^1|||^2_{\varepsilon,s}=&\|\phi^1\|^2_{H^s}+\varepsilon\|\phi^1\|^2_{H^s} +\varepsilon^2\|\Delta\phi^1\|^2_{H^s},
\end{split}
\end{equation}
and then show that $|||(n^1,{\bf u}^1,\phi^1)|||_{\varepsilon,s}$ is uniformly bounded on some time interval independent of $\varepsilon$. The main novelty of the proof is then to combine the $s$-order energy estimates with the $(s+1)$-order energy estimates with weights 1 and $\varepsilon$. By such a combination, we obtain some Gronwall type inequality for $|||(n^1,{\bf u}^1,\phi^1)|||_{\varepsilon,s}$, which enables us to obtain uniform estimates independent of $\varepsilon$. This method could be useful in treating the quasineutral limit for the pressureless electron Euler-Poisson equations.

We introduce several notations. We let $L^p$ denote the usual Lebesgue space of $p$-th integrable functions normed by $\|\cdot\|_{L^p}$. When $p=2$, we usually use $\|\cdot\|$ instead of $\|\cdot\|_{L^2}$. The Sobolev space $H^s$, $s\in\Bbb Z$, $s\geq0$ is defined as $H^s(\Bbb R^d)=\{f(x):\sum_{|\alpha|\leq s}\|\partial^{\alpha}f\|^2<\infty\}$, where $\alpha=(\alpha_1,\cdots,\alpha_d)$ is a multi-index, $|\alpha|=\sum\alpha_i$ and $\partial^{\alpha}=\frac{\partial^{|\alpha|}}{\partial^{\alpha_1}_{x_1} \cdots\partial^{\alpha_d}_{x_d}}$. $H^s$ is a Banach space with norm $\|f\|_{H^s}=(\sum_{|\alpha|\leq s}\|\partial^{\alpha}f\|^2)^{1/2}.$ For definiteness, we will restrict ourselves to the physical space dimensions $d\leq 3$ in this paper.

\section{Proof of Theorem \ref{thm}}
\setcounter{equation}{0}
The purpose of this section is to prove Theorem \ref{thm}. Let $(n^{\varepsilon},{\bf u}^{\varepsilon},\phi^{\varepsilon})$ satisfy the Euler-Poisson equation \eqref{EP}, and $(n^0,{\bf u}^0,\phi^0)$ be a sufficiently smooth solution of the Euler equation \eqref{EQ}. We let
\begin{equation}\label{Expan}
n^{\varepsilon}=n^0+\varepsilon n^1,\ \ {\bf u}^{\varepsilon}={\bf u}^0+\varepsilon {\bf u}^1,\ \ \phi^{\varepsilon}=\phi^0+\varepsilon\phi^1.
\end{equation}
Here $\phi^{\varepsilon}$ and $n^{\varepsilon}$ satisfy the Poisson equation \eqref{EP-p} and indeed $\phi^{\varepsilon}$ can be solved via $\phi^{\varepsilon}=\phi^{\varepsilon}[n^{\varepsilon}]$ and ${\phi^0}=\ln n^0$. Then $(n^1,{\bf u}^1,\phi^1)$ satisfy the remainder system ($R_{\varepsilon}$):
\begin{subequations}\label{rem}
\begin{numcases}{(R_{\varepsilon})\ \ }
\partial_tn^1+\nabla\cdot(n^0{\bf u}^1+{\bf u}^0n^1)+\varepsilon\nabla\cdot(n^1{\bf u}^1)=0\label{rem-n}\\
\partial_t{\bf u}^1+{\bf u}^0\cdot\nabla {\bf u}^1+{\bf u}^1\cdot\nabla {\bf u}^0+\varepsilon {\bf u}^1\cdot\nabla {\bf u}^1=-\nabla\phi^1\label{rem-u}\\
-\varepsilon\Delta\phi^1=\Delta\phi^0+n^1-n^0\phi^1+\sqrt\varepsilon R^1,\label{rem-p}
\end{numcases}
\end{subequations}
where
\begin{equation}\label{R1}
R^1=\varepsilon^{-3/2}(n^0+\varepsilon n^0\phi^1-e^{\phi^0+\varepsilon\phi^1}).
\end{equation}

To prove Theorem \ref{thm}, we need only to derive some uniform bound for the remainder equation \eqref{rem}. To slightly simplify the presentation, we assume that \eqref{rem} has smooth solutions in a small time $T_{\varepsilon}$ dependent on $\varepsilon$. Let $\tilde C$ be a constant to be determined later, much larger than the bound of $\|(n^{1}_0,\textbf{u}^{1}_0,\phi^{1}_0)\|_{s}$, such that on $[0,T_{\varepsilon}]$
\begin{equation}\label{assumption}
\begin{split}
\sup_{[0,T_{\varepsilon}]}\|(n^{1},\textbf{u}^{1},\phi^{1})\|_{H^s}\leq \tilde C.
\end{split}
\end{equation}
We will prove that $T_{\varepsilon}>T$ as $\varepsilon\to 0$ for some $T>0$. Recalling the expressions for $n$ and $\textbf{u}$ in \eqref{Expan}, we immediately know that there exists some $\varepsilon_1=\varepsilon_1(\tilde C)>0$ such that on $[0,T_{\varepsilon}]$,
\begin{equation}\label{assumption1}
\begin{split}
\sigma'/2<n^{\varepsilon}<2\sigma'',\ \ \ |{\bf u}^{\varepsilon}|\leq 1/2,
\end{split}
\end{equation}
for all $0<\varepsilon<\varepsilon_1$.

\subsection{Estimates for $R_1$}
We first bound $R^1$ in terms of $\phi^1$. More precisely, we have the following
\begin{lemma}\label{lem1}
Let $(n^0,{\bf u}^0,\phi^0)$ be a sufficiently smooth solution of \eqref{EQ} by Theorem \ref{thm1}. Then for the remainder term \eqref{R1}, we have on $[0,T_{\varepsilon}]$
\begin{equation}\label{equ19}
\begin{split}
\|R^1\|_{H^{k}}\leq & C(\sqrt{\varepsilon}\tilde C)\|\phi^1\|_{H^{k}},\ \ \ and\\
\|\partial_tR^1\|_{H^k}\leq & C(\sqrt{\varepsilon}\tilde C)(\|\phi^1\|_{H^k}+\|\partial_t\phi^1\|_{H^{k}}),\ \ \forall k\geq0.
\end{split}
\end{equation}
In particular, there exists some $\varepsilon_1>0$ and $C_1=C(1)$ such that
\begin{equation}\label{equ20}
\begin{split}
\|R^1\|_{H^{k}}\leq &C_1\|\phi^1\|_{H^{k}},\ \ \ and\\
\|\partial_tR^1\|_{H^k}\leq & C_1(\|\phi^1\|_{H^k}+\|\partial_t\phi^1\|_{H^{k}}),\ \ \forall k\geq0,
\end{split}
\end{equation}
\end{lemma}
for all $0<\varepsilon<\varepsilon_1$ and $t\in[0,T_{\varepsilon}]$.
\begin{proof}
From the Taylor expansion in the integral form, we have
\begin{equation*}
\begin{split}
R^1
=&\varepsilon^{1/2}e^{\phi^0}\int_0^1e^{\theta\varepsilon\phi^1} (1-\theta)d\theta(\phi^1)^2.
\end{split}
\end{equation*}
By taking $L^2$ norm, we have
\begin{equation*}
\|R^1\|\leq \sqrt{\varepsilon}\|e^{\phi^0}\|_{L^{\infty}}e^{\varepsilon \|\phi^1\|_{L^{\infty}}}\|\phi^1\|_{L^2}\|\phi^1\|_{L^{\infty}}.
\end{equation*}
From the continuity assumption \eqref{assumption}, we have $\|\phi^1\|_{L^{\infty}}\leq C\tilde C$ on $[0,T_{\varepsilon}]$ and hence
\begin{equation*}
\|R^1\|_{L^2}\leq C(\sqrt{\varepsilon}\tilde C)\|\phi^1\|_{L^2},\ \ \ \forall t\in [0,T_{\varepsilon}].
\end{equation*}


By applying $\partial^{\alpha}$ with $|\alpha|=k$, $k\geq 1$ integers, similar estimates yield
\begin{equation*}
\|R^1\|_{H^k}\leq C(\sqrt{\varepsilon}\tilde C)\|\phi^1\|_{H^k}.
\end{equation*}
Taking $\partial_t$ to $R^1$ and then taking the $H^k$ norm, we obtain
\begin{equation*}
\|\partial_tR^1\|_{H^k}\leq C(\sqrt{\varepsilon}\tilde C)(\|\phi^1\|_{H^k}+\|\partial_t\phi^1\|_{H^{k}}).
\end{equation*}
Finally, choosing $\varepsilon_1=(1/\tilde C)^2$ yields \eqref{equ20}.
\end{proof}

\subsection{Elliptical estimates}
The following lemmas provide useful estimates between $n^1$, ${\bf u}^1$ and $\phi^1$. These will be used widely in the uniform estimates in the next subsection.
\begin{lemma}\label{lem2}
Let $(n^1,{\bf u}^1,\phi^1)$ be a smooth solution for the remainder system ($R_{\varepsilon}$), and $\alpha$ be a multiindex. There exist $\varepsilon_1$ and $C$ such that for any $0<\varepsilon<\varepsilon_1$ and any multiindices $\alpha$ with $|\alpha|=k\geq0$, there hold
\begin{equation*}
\begin{split}
&\|\partial_x^{\alpha}n^1\|^2\leq C+C\|\phi^1\|_{H^k}^2+C\varepsilon^2\|\Delta\phi^1\|_{H^k}^2,\ \ \ and\\
&\|\partial_x^{\alpha}\phi^1\|^2 +\varepsilon\|\partial_x^{\alpha}\nabla\phi^1\| +\varepsilon^2\|\Delta\partial_x^{\alpha}\phi^1\|^2\leq C+C\|\partial_x^{\alpha}n^1\|^2,
\end{split}
\end{equation*}
on the interval $[0,T_{\varepsilon}]$.
\end{lemma}
\begin{proof}
Taking the $L^2$ inner product of \eqref{rem-p} with $\phi^1$ and then integrating by parts yield
\begin{equation*}
\varepsilon\|\nabla\phi^1\|^2+\int n^0|\phi^1|^2=\int \phi^1\Delta\phi^0+\int n^1\phi^1+\sqrt\varepsilon\int\phi^1R^1.
\end{equation*}
Hereafter, $\int=\int_{\Bbb R^d}\cdots dx$. As $n^0>\sigma'>0$ for $t\leq T'<T$, we obtain by Young's inequality
\begin{equation*}
\varepsilon\|\nabla\phi^1\|^2+\sigma'\|\phi^1\|^2\leq \frac{\sigma'}{4}\|\phi^1\|^2 +\frac{4C}{\sigma'}(\|\Delta\phi^0\|^2+\|n^1\|^2+\varepsilon\|R^1\|^2).
\end{equation*}
From Lemma \ref{lem1}, there exists $\varepsilon_1>0$ such that $\|R^1\|\leq C_1\|\phi^1\|$. Then by choosing a new smaller $\varepsilon_1$ such that $\varepsilon_1\leq \frac{\sigma'^2}{16CC_1}$, we then have for any $0<\varepsilon<\varepsilon_1$ that
\begin{equation}\label{equ1}
\varepsilon\|\nabla\phi^1\|^2+\sigma'\|\phi^1\|^2\leq \frac{C}{\sigma'}(1+\|n^1\|^2).
\end{equation}
Similarly, by taking the $L^2$ inner product of \eqref{rem-p} with $\varepsilon\Delta\phi^1$ and integrating by parts, we obtain that
\begin{equation*}
\begin{split}
\varepsilon^2\|\Delta\phi^1\|^2&+\varepsilon\int n^0|\nabla\phi^1|^2 =-\varepsilon\int\Delta\phi^0\Delta\phi^1+\varepsilon\int n^1\Delta\phi^1\\
&-\varepsilon\int \nabla n^0\phi^1\nabla\phi^1+\varepsilon^{3/2}\int R^1\phi^1\\
\leq & \frac{\varepsilon^2}{2}\|\Delta\phi^1\|^2+ 4(\|\Delta\phi^0\|^2+\|n^1\|^2+\varepsilon\|R^1\|^2 +\varepsilon\|\phi^1\|^2+\varepsilon\|\nabla\phi^1\|^2),
\end{split}
\end{equation*}
which yields for any $0<\varepsilon<\varepsilon_1$ for some $\varepsilon_1>0$ that
\begin{equation}\label{equ2}
\varepsilon^2\|\Delta\phi^1\|^2+\varepsilon\sigma'\|\nabla\phi^1\|^2\leq C(1+\|n^1\|^2),
\end{equation}
as $n^0>\sigma'>0$ for $t\leq T'<T$, where the constant $C$ depends on $\sigma'$. By combining \eqref{equ1} and \eqref{equ2} together, we easily obtain that for any $0<\varepsilon<\varepsilon_1$:
\begin{equation*}
\varepsilon^2\|\Delta\phi^1\|^2+\varepsilon\sigma'\|\nabla\phi^1\|^2 +\|\phi^1\|^2\leq C(1+\|n^1\|^2),
\end{equation*}
for some constant $C$ depending on $\sigma'$. On the other hand, by taking the $L^2$ norm of \eqref{rem-p}, we obtain that for any $0<\varepsilon<\varepsilon_1$:
\begin{equation*}
\begin{split}
\|n^1\|^2\leq & \varepsilon^2\|\Delta\phi^1\|^2 +\|n^0\|_{L^{\infty}}^2\|\phi^1\|^2+\|\Delta\phi^0\|^2+\varepsilon\|R^1\|^2\\
\leq & C(1+\|\phi^1\|^2+\varepsilon^2\|\Delta\phi^1\|^2),
\end{split}
\end{equation*}
thanks again to Lemma \ref{lem1}, where $C$ depends on $\sigma''$ and $C_1$. Therefore, we finishes the proof when $k=0$. Higher order estimates can be handled similarly, and we omit further details.
\end{proof}

\begin{lemma}\label{lem3}
Let $(n^1,{\bf u}^1,\phi^1)$ be a smooth solution for the remainder system ($R_{\varepsilon}$), and $\alpha$ be an integer. There exist $\varepsilon_1$ and $C$ such that for any $0<\varepsilon<\varepsilon_1$ and any multiindices $\alpha$ with $|\alpha|=k$, there holds
\begin{equation}\label{equ3}
\begin{split}
\|\partial^{\alpha}\partial_tn^1\|^2\leq C(1+\|{\bf u}^1\|_{{\bf H}^{k+1}}^2 +\|\phi^1\|^2_{H^{k+1}}+\varepsilon^2\|\Delta\phi^1\|^2_{H^{k+1}}),
\end{split}
\end{equation}
on the time interval $[0,T_{\varepsilon}]$.
\end{lemma}
\begin{proof}
We take the $L^2$ norm of \eqref{rem-n} to obtain
\begin{equation*}
\begin{split}
\|\partial_tn^1\|^2\leq & C({\bf u}^1\|_{{\bf H}^1}^2+\|{n}^1\|_{H^1}^2) +\varepsilon^2(\|n^1\|_{L^{\infty}}^2\|\nabla{\bf u}^1\|^2 +\|{\bf u}^1\|_{L^{\infty}}^2\|\nabla{n}^1\|^2)\\
\leq & C(1+\varepsilon^2(\|n^1\|_{L^{\infty}}^2+\|{\bf u}^1\|_{L^{\infty}}^2))(\|{\bf u}^1\|_{{\bf H}^1}^2+\|{n}^1\|_{H^1}^2),
\end{split}
\end{equation*}
for some constant $C$ depending on $(n^0,{\bf u}^0)$. By the continuity assumption \eqref{assumption} and Lemma \ref{lem2}, we have
\begin{equation*}
\begin{split}
\|\partial_tn^1\|^2\leq & C(1+\varepsilon^2\tilde C)(1+\|{\bf u}^1\|_{{\bf H}^1}^2+\|\phi^1\|^2_{H^1}+\varepsilon^2\|\Delta\phi^1\|^2_{H^1}),\\
\leq & C(1+\|{\bf u}^1\|_{{\bf H}^1}^2+\|\phi^1\|^2_{H^1}+\varepsilon^2\|\Delta\phi^1\|^2_{H^1})
\end{split}
\end{equation*}
for any $0<\varepsilon<\varepsilon_1$, for some $\varepsilon_1>0$.

Higher order inequalities are proved similarly. Taking $\partial^{\alpha}$ with $|\alpha|=k\geq1$ to the equation \eqref{rem-n}, and then taking the $L^2$ norm to obtain
\begin{equation*}
\begin{split}
\|\partial^{\alpha}\partial_tn^1\|^2\leq & C({\bf u}^1\|_{{\bf H}^{k+1}}^2+\|{n}^1\|_{H^{k+1}}^2) +\varepsilon^2(\|n^1\|_{L^{\infty}}^2\|\partial^{\alpha}\nabla{\bf u}^1\|^2 +\|{\bf u}^1\|_{L^{\infty}}^2\|\partial^{\alpha}\nabla{n}^1\|^2)\\
\leq & C(1+\varepsilon^2(\|n^1\|_{L^{\infty}}^2+\|{\bf u}^1\|_{L^{\infty}}^2))(\|{\bf u}^1\|_{{\bf H}^{k+1}}^2+\|{n}^1\|_{H^{k+1}}^2)\\
\leq & C(1+\|{\bf u}^1\|_{{\bf H}^{k+1}}^2+\|\phi^1\|^2_{H^{k+1}}+\varepsilon^2\|\Delta\phi^1\|^2_{H^{k+1}}),
\end{split}
\end{equation*}
for any $0<\varepsilon<\varepsilon_1$, for some $\varepsilon_1>0$, where we have used the multiplicative estimates in Lemma \ref{Le-inequ}.
\end{proof}

\begin{lemma}\label{lem4}
Let $(n^1,{\bf u}^1,\phi^1)$ be a smooth solution for the remainder system ($R_{\varepsilon}$), and $\alpha$ be an integer. There exist $\varepsilon_1$ and $C$ such that for any $0<\varepsilon<\varepsilon_1$ and any multiindices $\alpha$ with $|\alpha|=k$,
\begin{equation*}
\begin{split}
\|\partial^{\alpha}\partial_t\phi^1\|^2 +\varepsilon\|\partial^{\alpha}\nabla\partial_t\phi^1\|^2 +\varepsilon^2\|\partial^{\alpha}\Delta\partial_t\phi^1\|^2 \leq C(1+\|\partial^{\alpha}\partial_tn^1\|^2 +\|\phi^1\|_{H^k}^2),
\end{split}
\end{equation*}
on the time interval $[0,T_{\varepsilon}]$.
\end{lemma}
\begin{proof}
Taking $\partial_t$ of \eqref{rem-p}, we obtain
\begin{equation*}
-\varepsilon\Delta\partial_t\phi^1=\Delta\partial_t\phi^0+\partial_tn^1 -\partial_t(n^0\phi^1)+\sqrt\varepsilon\partial_tR^1.
\end{equation*}
Taking $L^2$ inner product with $\partial_t\phi^1$ and then integrating by parts, we obtain
\begin{equation*}
\begin{split}
\varepsilon\|\nabla\partial_t\phi^1\|^2+\int n^0|\partial_t\phi^1|^2=&\int\partial_t\phi^1(\Delta\partial_t\phi^0 -\partial_tn^0\phi^1+\partial_tn^1+\sqrt\varepsilon\partial_tR^1).
\end{split}
\end{equation*}
As $n^0>\sigma'>0$ for $t\leq T'<T$, we have by H\"older inequality
\begin{equation*}
\begin{split}
\varepsilon\|\nabla\partial_t\phi^1\|^2&+\sigma'\|\partial_t\phi^1\|^2 \leq \frac{\sigma'}{4}\|\partial_t\phi^1\|^2\\
&+\frac{4}{\sigma'}\big(\|\Delta\partial_t\phi^0\|^2 +\|\partial_tn^0\|_{L^{\infty}}^2\|\phi^1\|^2+\|\partial_tn^1\|^2 +\varepsilon\|\partial_tR^1\|^2\big)\\
\leq & \frac{\sigma'}{4}\|\partial_t\phi^1\|^2 +\frac{4}{\sigma'} (C+C\|\phi^1\|^2+\|\partial_tn^1\|^2 +\varepsilon\|\partial_tR^1\|^2).
\end{split}
\end{equation*}
By choosing a small $\varepsilon_1>0$ such that $16C_1\varepsilon_1\leq \sigma'^2$, we then have for any $0<\varepsilon<\varepsilon_1$ that
\begin{equation*}
\varepsilon\|\nabla\partial_t\phi^1\|^2+\sigma'\|\partial_t\phi^1\|^2\leq \frac{C}{\sigma'}(1+\|\phi^1\|^2+\|\partial_tn^1\|^2),
\end{equation*}
thanks to Lemma \ref{lem1}. Similarly, by taking inner product with $\varepsilon\Delta\partial_t\phi^1$, we obtain
\begin{equation*}
\varepsilon^2\|\Delta\partial_t\phi^1\|^2+\sigma'\varepsilon\|\partial_t\phi^1\|^2\leq \frac{C}{\sigma'}(1+\|\phi^1\|^2+\|\partial_tn^1\|^2).
\end{equation*}
Adding them together, we obtain that for and any $0<\varepsilon<\varepsilon_1$ for some $\varepsilon_1>0$,
\begin{equation*}
\begin{split}
\|\partial_t\phi^1\|^2 +\varepsilon\|\nabla\partial_t\phi\|^2 +\varepsilon^2\|\Delta\partial_t\phi^1\|^2\leq C(1+\|\partial_tn^1\|^2_{H^k} +\|\phi^1\|_{H^k}^2),
\end{split}
\end{equation*}
for some constant $C$ depending on $\sigma'$. Higher order estimates can be treated similarly and we obtain for any $\alpha$ with $|\alpha|=k$ that
\begin{equation*}
\begin{split}
\|\partial_t\partial^{\alpha}\phi^1\|^2 +\varepsilon\|\partial_t\partial^{\alpha}\nabla\phi\|^2 +\varepsilon^2\|\partial_t\partial^{\alpha}\Delta\phi^1\|^2\leq C(1+\|\partial_tn^1\|^2_{H^k} +\|\phi^1\|_{H^k}^2),
\end{split}
\end{equation*}
for some constant $C$ depending on $\sigma'$.
\end{proof}

By recalling Lemma \ref{lem3}, we have the following
\begin{corollary}\label{cor2}
Let $(n^1,{\bf u}^1,\phi^1)$ be a smooth solution for the remainder system ($R_{\varepsilon}$), and $\alpha$ be an integer. There exist $\varepsilon_1$ and $C$ such that for any $0<\varepsilon<\varepsilon_1$ and any multiindices $\alpha$ with $|\alpha|=k$,
\begin{equation*}
\begin{split}
\|\partial^{\alpha}\partial_t\phi^1\|^2 &+\varepsilon\|\partial^{\alpha}\nabla\partial_t\phi^1\|^2 +\varepsilon^2\|\partial^{\alpha}\Delta\partial_t\phi^1\|^2\\
&\leq C(1+\|{\bf u}^1\|_{{\bf H}^{k+1}}^2 +\|\phi^1\|^2_{H^{k+1}}+\varepsilon^2\|\Delta\phi^1\|^2_{H^{k+1}}),
\end{split}
\end{equation*}
on the time interval $[0,T_{\varepsilon}]$.
\end{corollary}

As a direct consequence of Lemma \ref{lem1}, Lemma \ref{lem4} and \ref{lem3}, we also have
\begin{corollary}\label{cor1}
Let $(n^1,{\bf u}^1,\phi^1)$ be a smooth solution for the remainder system ($R_{\varepsilon}$), and $\alpha$ be an integer. There exist $\varepsilon_1$ and $C$ such that
\begin{equation*}
\begin{split}
\|\partial_tR^1\|_{H^k}\leq  C(1+|||({\bf u}^1,\phi^1)|||_{\varepsilon,k+1}),
\end{split}
\end{equation*}
for any $0<\varepsilon<\varepsilon_1$ and any multiindices $\alpha$ with $|\alpha|=k$.
\end{corollary}
\subsection{Estimates of the $s$ order}
In this subsection, we give several estimates at the $s$ order. However, the $H^s$-norm of the solutions depends on the $H^{s+1}$-norm and hence cannot be closed until the next subsection. The main result in this subsection is Proposition \ref{prop1}. In the following, $\gamma\geq0$ will always denote a multiindex with $|\gamma|=s$.
\begin{lemma}\label{lem5}
Let $\gamma\geq0$ be a multiindex with $|\gamma|=s$, $(n^1,{\bf u}^1,\phi^1)$ be a smooth solution for the system \eqref{rem}. There exists $\varepsilon_1>0$ and $C>0$ such that
\begin{equation}\label{equ101}
\begin{split}
\{\frac12\frac{d}{dt}\|\partial^{\gamma}{\bf u}^1\|^2_{L^2} &+\frac12\frac{d}{dt}\int\frac{n^0}{n^0+\varepsilon n^1}|\partial^{\gamma}\phi^1|^2 +\frac{\varepsilon}{2}\frac{d}{dt}\int\frac{1}{n^0+\varepsilon n^1} |\partial^{\gamma}\nabla\phi^1|^2\} \\
\leq & C(1+\varepsilon^3|||({\bf u}^1,\phi^1)|||^3_{\varepsilon,5})(1+|||({\bf u}^1,\phi^1)|||^2_{\varepsilon,s}+|||({\bf u}^1,\phi^1)|||^2_{\varepsilon,3}),
\end{split}
\end{equation}
for any $0<\varepsilon<\varepsilon_1$.
\end{lemma}
\begin{proof}
Let $\gamma$ be a multiindex with $|\gamma|=s\geq0$. Taking $\partial^{\gamma}$ to \eqref{rem-u}, we obtain
\begin{equation*}
\partial_t\partial^{\gamma}{\bf u}^{1}+\partial^{\gamma}({\bf u}^0\cdot\nabla {\bf u}^1)+\partial^{\gamma}({\bf u}^1\cdot\nabla {\bf u}^0)+\varepsilon \partial^{\gamma}({\bf u}^1\cdot\nabla {\bf u}^1)=-\partial^{\gamma}\nabla\phi^1.
\end{equation*}
Taking $L^2$ inner product with $\partial^{\gamma}{\bf u}^1$, we obtain
\begin{equation}\label{equ4}
\begin{split}
\int\partial_t\partial^{\gamma}{\bf u}^1\partial^{\gamma}{\bf u}^1 =&-\int\partial^{\gamma}\nabla\phi^1\partial^{\gamma}{\bf u}^1 -\varepsilon\int\partial^{\gamma}({\bf u}^1\cdot\nabla {\bf u}^1)\partial^{\gamma}{\bf u}^1\\
&-\int\partial^{\gamma}({\bf u}^0\cdot\nabla {\bf u}^1)\partial^{\gamma}{\bf u}^1 -\int\partial^{\gamma}({\bf u}^1\cdot\nabla {\bf u}^0)\partial^{\gamma}{\bf u}^1\\
=&:I+II+III+IV.
\end{split}
\end{equation}

\begin{itemize}
  \item \emph{Estimate of the fourth term $IV$.}
\end{itemize}
The term $IV$ can be bounded by
\begin{equation}\label{equ30}
\begin{split}
IV\leq & C\|\partial^{\gamma}({\bf u}^1\cdot\nabla{\bf u}^0)\|_{L^2}\|\partial^{\gamma}{\bf u}^1\|_{L^2}\\
\leq & C(\|{\bf u}^1\|_{H^s}\|\nabla{\bf u}^0\|_{L^{\infty}}+\|\nabla{\bf u}^0\|_{H^s}\|{\bf u}^1\|_{L^{\infty}})\|{\bf u}^1\|_{H^s}\\
\leq & C(\|{\bf u}^1\|^2_{H^s}+\|{\bf u}^1\|^2_{H^2}),
\end{split}
\end{equation}
where we have used the commutator estimates \eqref{commutator}, the Sobolev embedding $H^2\hookrightarrow L^{\infty}$ when $d\leq 3$ and the fact that $(n^0,{\bf u}^0)$ is a known smooth solution of the Euler equation \eqref{EQ} by Theorem \ref{thm1}.

\begin{itemize}
  \item \emph{Estimate of the third term $III$.}
\end{itemize}
By integration by parts, the third term $III$ can be rewritten as
\begin{equation}\label{equ21}
\begin{split}
III=&-\int{\bf u}^0\cdot\nabla\partial^{\gamma}{\bf u}^1\partial^{\gamma}{\bf u}^1-\int[\partial^{\gamma},{\bf u}^0]\cdot\nabla{\bf u}^1\partial^{\gamma}{\bf u}^1\\
=& \frac12\int\nabla\cdot{\bf u}^0|\partial^{\gamma}{\bf u}^1|^2 -\int[\partial^{\gamma},{\bf u}^0]\cdot\nabla{\bf u}^1\partial^{\gamma}{\bf u}^1.
\end{split}
\end{equation}
By using commutator estimates \eqref{commutator}, we obtain
\begin{equation*}
\begin{split}
\left|\int[\partial^{\gamma},{\bf u}^0]\cdot\nabla{\bf u}^1\partial^{\gamma}{\bf u}^1\right|\leq & C\|[\partial^{\gamma},{\bf u}^0]\cdot\nabla{\bf u}^1\|_{L^2}\|\partial^{\gamma}{\bf u}^1\|_{L^2}\\
\leq & C(\|\nabla{\bf u}^0\|_{L^{\infty}}\|\nabla{\bf u}^1\|_{H^{s-1}}+\|{\bf u}^0\|_{H^s}\|\nabla{\bf u}^1\|_{L^{\infty}})\|\partial^{\gamma}{\bf u}^1\|_{L^2}\\
\leq & C(\|{\bf u}^1\|^2_{H^{3}}+\|{\bf u}^1\|^2_{H^{s}}),
\end{split}
\end{equation*}
where we have used the H\"older inequality and the Sobolev embedding $H^2\hookrightarrow L^{\infty}$. This yields the estimate
\begin{equation}\label{equ31}
\begin{split}
III\leq C(\|{\bf u}^1\|^2_{H^{3}}+\|{\bf u}^1\|^2_{H^{s}}),
\end{split}
\end{equation}
since the first term on the RHS of \eqref{equ21} is bounded by $C\|{\bf u}^1\|^2_{H^{s}}$.

\begin{itemize}
  \item \emph{Estimate of the second term $II$.}
\end{itemize}
Similar to the estimate of $III$, we have by integration by parts that
\begin{equation}\label{equ26}
\begin{split}
II= & \frac\varepsilon2\int\nabla\cdot{\bf u}^1\partial^{\gamma}{\bf u}^1\partial^{\gamma}{\bf u}^1 -\varepsilon\int[\partial^{\gamma},{\bf u}^1]\cdot\nabla{\bf u}^1\partial^{\gamma}{\bf u}^1\\
\leq & C\varepsilon\|\nabla\cdot{\bf u}^1\|_{L^{\infty}}\|\partial^{\gamma}{\bf u}^1\|^2_{L^2}\\
&+C\varepsilon(\|\nabla{\bf u}^1\|_{L^{\infty}}\|\nabla{\bf u}^1\|_{H^{s-1}}+\|\nabla{\bf u}^1\|_{L^{\infty}}\|{\bf u}^1\|_{H^s})\|\partial^{\gamma}{\bf u}^1\|_{L^2}\\
\leq & C\varepsilon\|{\bf u}^1\|_{H^3}\|{\bf u}^1\|^2_{H^{s}}.
\end{split}
\end{equation}

\begin{itemize}
  \item \emph{Estimate of the first term $I$.}
\end{itemize}
By integration by parts, the term $I$ in \eqref{equ4} is rewritten as
\begin{equation*}
\begin{split}
I=&\int\partial^{\gamma}\phi^1\partial^{\gamma}\nabla\cdot{\bf u}^1.
\end{split}
\end{equation*}
To handle this term, we note that from the remainder equation \eqref{rem-n},
\begin{equation}\label{equ13}
\begin{split}
(n^0+&\varepsilon n^1)\partial^{\gamma}\nabla\cdot{\bf u}^1 +[\partial^{\gamma},n^0+\varepsilon n^1]\nabla\cdot{\bf u}^1+\partial_t\partial^{\gamma}n^1\\
&+\partial^{\gamma}(({\bf u}^0+\varepsilon{\bf u}^1)\cdot\nabla n^1+{\bf u}^1\cdot\nabla n^0) +\partial^{\gamma}(n^1\nabla\cdot {\bf u}^0)=0.
\end{split}
\end{equation}
Inserting this into $I$, we obtain
\begin{equation}\label{equ6}
\begin{split}
I=& -\int\frac{\partial^{\gamma}\phi^1}{n^0+\varepsilon n^1}\partial_t\partial^{\gamma}n^1 -\int\frac{\partial^{\gamma}\phi^1}{n^0+\varepsilon n^1}\partial^{\gamma}(({\bf u}^0+\varepsilon{\bf u}^1)\cdot\nabla n^1)\\
&-\int\frac{\partial^{\gamma}\phi^1}{n^0+\varepsilon n^1}[\partial^{\gamma},n^0+\varepsilon n^1]\nabla\cdot{\bf u}^1 -\int\frac{\partial^{\gamma}\phi^1}{n^0+\varepsilon n^1}\partial^{\gamma}({\bf u}^1\cdot\nabla n^0)\\
&-\int\frac{\partial^{\gamma}\phi^1}{n^0+\varepsilon n^1}\partial^{\gamma}({n}^1\cdot\nabla {\bf u}^0) =:\sum_{i=1}^5I_{i}.
\end{split}
\end{equation}

In the following, we estimate $I_{3}\sim I_{5}$ while leaving the estimates of $I_{1}$ and $I_{2}$ to the next lemmas.

For $I_{3}$, we have
\begin{equation*}
\begin{split}
I_{3}\leq & C\|\partial^{\gamma}\phi^1\|_{L^2}(\|\nabla(n^0+\varepsilon n^1)\|_{L^{\infty}}\|\nabla\cdot{\bf u}^1\|_{H^{s-1}}+\|n^0+\varepsilon n^1\|_{H^{s}}\|\nabla\cdot{\bf u}^1\|_{L^{\infty}})\\
\leq & C\|\phi^1\|_{H^s}(\|{\bf u}^1\|_{H^s}+\|{\bf u}^1\|_{H^3}) +C\varepsilon(\|n^1\|_{H^3}+\|{\bf u}^1\|_{H^3})\|\phi^1\|_{H^s}(\|n^1\|_{H^s}+\|{\bf u}^1\|_{H^s})\\
\leq & C(1+\varepsilon\|(n^1,{\bf u}^1)\|_{H^3})\|(n^1,{\bf u}^1,\phi^1)\|_{H^s}^2+C\|{\bf u}^1\|_{H^3}^2,
\end{split}
\end{equation*}
where we have used the H\"older inequality, commutator estimates \eqref{commutator} and the fact that $n^0$ and $n^0+\varepsilon n^1$ are bounded from above and below by positive numbers when $\varepsilon<\varepsilon_1$ is small enough in \eqref{assumption1}.

For $I_{4}$, directly applying the H\"older inequality and Lemma \ref{Le-inequ} yields
\begin{equation*}
\begin{split}
I_{24}\leq & C\|\phi^1\|_{H^s}(\|{\bf u}^1\|_{L^\infty}\|\nabla n^0\|_{H^s}+\|{\bf u}^1\|_{H^s}\|\nabla n^0\|_{L^\infty})\\
\leq & C\|\phi^1\|^2_{H^s}+C\|{\bf u}^1\|^2_{H^s}+C\|{\bf u}^1\|^2_{H^2}.
\end{split}
\end{equation*}

Similarly, $I_{5}$ can be bounded by
\begin{equation*}
\begin{split}
I_{25}\leq C\|\phi^1\|^2_{H^s}+C\|n^1\|^2_{H^s}+C\|n^1\|^2_{H^2}.
\end{split}
\end{equation*}

Summarizing, we have that
\begin{equation}\label{equ29}
\begin{split}
I\leq & I_{1}+I_{2}+C\|(n^1,{\bf u}^1)\|^2_{H^3} +C(1+\varepsilon^2\|(n^1,{\bf u}^1)\|^2_{H^3})\|(n^1,{\bf u}^1,\phi^1)\|^2_{H^s}.
\end{split}
\end{equation}

To end the proof of Lemma \ref{lem5}, we need to get suitable estimates for $I_{1}$ and $I_{2}$. However, this is not straightforward and to make it easier to read, we leave the proof to the next two lemmas.
\end{proof}

\begin{lemma}\label{lem6}
The term of $I_{2}$ in \eqref{equ6} is bounded by
\begin{equation}\label{equ5}
\begin{split}
I_{21}\leq C(1+\varepsilon^3|||({\bf u}^1,\phi^1)|||^3_{\varepsilon,4}) (1+|||({\bf u}^1,\phi^1)|||^2_{\varepsilon,s}+|||({\bf u}^1,\phi^1)|||^2_{\varepsilon,3}),
\end{split}
\end{equation}
for some constant $C>0$ and for all $0<\varepsilon<\varepsilon_1$.
\end{lemma}
\begin{proof}
First, we observe that $I_{2}$ in \eqref{equ6} can be decomposed into
\begin{equation*}
\begin{split}
I_{2}=-\int\frac{\partial^{\gamma}\phi^1({\bf u}^0+\varepsilon{\bf u}^1)}{n^0+\varepsilon n^1}\cdot\nabla\partial^{\gamma} n^1 -\int\frac{\partial^{\gamma}\phi^1}{n^0+\varepsilon n^1}[\partial^{\gamma},{\bf u}^0+\varepsilon{\bf u}^1]\cdot\nabla n^1=:I_{21}+I_{22}.
\end{split}
\end{equation*}
By commutator estimate \eqref{commutator},
\begin{equation}\label{equ27}
\begin{split}
I_{22}\leq & C\|\partial^{\gamma}\phi^1\|(\|\nabla{\bf u}^0+\varepsilon\nabla{\bf u}^1\|_{L^{\infty}}\|\nabla n^1\|_{H^{s-1}} +\|{\bf u}^0+\varepsilon{\bf u}^1\|_{H^s}\|\nabla n^1\|_{L^{\infty}})\\
\leq & C\|\phi^1\|_{H^s}(\|n^1\|_{H^s}+\|n^1\|_{H^3}) +C\varepsilon(\|(n^1,{\bf u}^1)\|_{H^3})\|\phi^1\|_{H^s}(\|(n^1,{\bf u}^1)\|_{H^s})\\
\leq & C\|n^1\|^2_{H^3} +C(1+\varepsilon\|(n^1,{\bf u}^1)\|_{H^3})\|(n^1,{\bf u}^1,\phi^1)\|^2_{H^s}.
\end{split}
\end{equation}
To treat $I_{21}$, we first note that from the remainder equation \eqref{rem-p},
\begin{equation}\label{equ15}
\begin{split}
\partial^{\gamma}\nabla n^1=\nabla\partial^{\gamma}(n^0\phi^1) -\varepsilon\partial^{\gamma}\nabla\Delta\phi^1 -\partial^{\gamma}\nabla\Delta\phi^0 -\sqrt\varepsilon\partial^{\gamma}\nabla R^1.
\end{split}
\end{equation}
Hence $I_{21}$ is accordingly divided into
\begin{equation}\label{equ50}
\begin{split}
I_{21}=&-\int\frac{\partial^{\gamma}\phi^1({\bf u}^0+\varepsilon{\bf u}^1)}{n^0+\varepsilon n^1}\cdot\nabla\partial^{\gamma}(n^0\phi^1) +\varepsilon\int\frac{\partial^{\gamma}\phi^1({\bf u}^0+\varepsilon{\bf u}^1)}{n^0+\varepsilon n^1}\cdot\partial^{\gamma}\nabla\Delta\phi^1\\
&+\int\frac{\partial^{\gamma}\phi^1({\bf u}^0+\varepsilon{\bf u}^1)}{n^0+\varepsilon n^1}\cdot\partial^{\gamma}\nabla\Delta\phi^0 +\sqrt\varepsilon\int\frac{\partial^{\gamma}\phi^1({\bf u}^0+\varepsilon{\bf u}^1)}{n^0+\varepsilon n^1}\cdot\partial^{\gamma}\nabla R^1\\
=&:\sum_{i=1}^4{I_{21i}}.
\end{split}
\end{equation}
For the first term $I_{211}$, we have
\begin{equation*}
\begin{split}
I_{211}=&-\int\frac{\partial^{\gamma}\phi^1({\bf u}^0+\varepsilon{\bf u}^1)}{n^0+\varepsilon n^1}\cdot n^0\nabla\partial^{\gamma}\phi^1 -\int\frac{\partial^{\gamma}\phi^1({\bf u}^0+\varepsilon{\bf u}^1)}{n^0+\varepsilon n^1}\cdot[\partial^{\gamma},n^0]\nabla\phi^1\\
=&\frac12\int\nabla\cdot\left(\frac{n^0({\bf u}^0+\varepsilon{\bf u}^1)}{n^0+\varepsilon n^1}\right)|\partial^{\gamma}\phi^1|^2 -\int\frac{\partial^{\gamma}\phi^1({\bf u}^0+\varepsilon{\bf u}^1)}{n^0+\varepsilon n^1}\cdot[\partial^{\gamma},n^0]\nabla\phi^1.
\end{split}
\end{equation*}
By direct computation and Sobolev embedding, we have
\begin{equation*}
\begin{split}
\|\nabla\cdot\left(\frac{n^0({\bf u}^0+\varepsilon{\bf u}^1)}{n^0+\varepsilon n^1}\right)\|_{L^{\infty}}\leq C+C\varepsilon^2(\|n^1\|^2_{H^{3}}+\|{\bf u}^1\|^2_{H^{3}}),
\end{split}
\end{equation*}
which yields
\begin{equation}\label{equ51}
\begin{split}
I_{211}\leq & C\|\phi^1\|^2_{H^s}+C\varepsilon^2(\|n^1\|^2_{H^3}+\|{\bf u}^1\|^2_{H^3})\|\phi^1\|^2_{H^s} +C\|\partial^{\gamma}\phi^1\|^2_{L^2}\cdot\\
&\cdot(1+\varepsilon\|{\bf u}^1\|_{L^{\infty}})(\|\nabla n^0\|_{L^{\infty}}\|\nabla\phi^1\|_{H^{s-1}} +\|n^0\|_{H^s}\|\nabla\phi^1\|_{L^{\infty}})\\
\leq & C(1+\varepsilon^2(\|n^1\|^2_{H^3}+\|{\bf u}^1\|^2_{H^3}))\|\phi^1\|^2_{H^s}+C\|\phi^1\|^2_{H^3},
\end{split}
\end{equation}
thanks to the commutator estimates \eqref{commutator}. For $I_{212}$, by integration by parts, we obtain
\begin{equation*}
\begin{split}
I_{212}=&\varepsilon\int\frac{({\bf u}^0+\varepsilon{\bf u}^1)\partial^{\gamma}\phi^1}{n^0+\varepsilon n^1}\cdot\partial^{\gamma}\nabla\Delta\phi^1\\
=&-\varepsilon\int\frac{({\bf u}^0+\varepsilon{\bf u}^1)}{n^0+\varepsilon n^1}\cdot\partial^{\gamma}\nabla\phi^1 \cdot\partial^{\gamma}\nabla\nabla\phi^1 -\varepsilon\int\nabla(\frac{({\bf u}^0+\varepsilon{\bf u}^1)}{n^0+\varepsilon n^1})\partial^{\gamma}\phi^1 \cdot\partial^{\gamma}\nabla\nabla\phi^1\\
=& \frac{3\varepsilon}{2}\int\nabla\cdot(\frac{({\bf u}^0+\varepsilon{\bf u}^1)}{n^0+\varepsilon n^1})|\partial^{\gamma}\nabla\phi^1|^2 +\varepsilon\int\nabla^2(\frac{({\bf u}^0+\varepsilon{\bf u}^1)}{n^0+\varepsilon n^1})\partial^{\gamma}\phi^1\partial^{\gamma}\nabla\phi^1.
\end{split}
\end{equation*}
By direct computation and Sobolev embedding $H^2\hookrightarrow L^{\infty}$, we have
\begin{equation*}
\begin{split}
\|\nabla\cdot(\frac{({\bf u}^0+\varepsilon{\bf u}^1)}{n^0+\varepsilon n^1})\|_{L^{\infty}}\leq &C+C\varepsilon^2(\|n^1\|^2_{H^{3}}+\|{\bf u}^1\|^2_{H^{3}})\\
\|\nabla^2(\frac{({\bf u}^0+\varepsilon{\bf u}^1)}{n^0+\varepsilon n^1})\|_{L^{\infty}}\leq &C+C\varepsilon^3(\|n^1\|^3_{H^{4}}+\|{\bf u}^1\|^3_{H^{4}}),
\end{split}
\end{equation*}
which yield
\begin{equation}\label{equ52}
\begin{split}
I_{212}\leq &C\left(1+\varepsilon^3(\|(n^1,{\bf u}^1)\|^3_{H^4})\right) (\|\phi^1\|^2_{H^s}+\varepsilon\|\nabla\phi^1\|^2_{H^s}).
\end{split}
\end{equation}
For $I_{213}$, by H\"older inequality, we obtain
\begin{equation}\label{equ53}
\begin{split}
I_{213}\leq &C\left(1+\|{\bf u}^1\|^2_{L^2}+\|\phi^1\|^2_{H^s}\right).
\end{split}
\end{equation}
For $I_{214}$, by integration by parts,
\begin{equation*}
\begin{split}
I_{214}=-\sqrt\varepsilon\int\nabla\cdot(\frac{({\bf u}^0+\varepsilon{\bf u}^1)}{n^0+\varepsilon n^1})\partial^{\gamma}\phi^1\partial^{\gamma}R^1 -\sqrt\varepsilon\int\frac{({\bf u}^0+\varepsilon{\bf u}^1)}{n^0+\varepsilon n^1}\cdot\nabla\partial^{\gamma}\phi^1\partial^{\gamma}R^1.
\end{split}
\end{equation*}
By using Lemma \ref{lem1}, we obtain
\begin{equation}\label{equ54}
\begin{split}
I_{214}\leq & C(1+\varepsilon^2\|(n^1,{\bf u}^1)\|^2_{H^3}) \|\phi^1\|^2_{H^s}\\
&+C(1+\varepsilon\|{\bf u}^1\|_{L^{\infty}}) (\|\phi^1\|^2_{H^s}+\varepsilon\|\nabla\phi^1\|^2_{H^s})\\
\leq &C(1+\varepsilon\|({\bf u}^1,n^1)\|_{H^3})(\|\phi^1\|^2_{H^s}+\varepsilon\|\nabla\phi^1\|^2_{H^s}).
\end{split}
\end{equation}
By \eqref{equ50} and putting \eqref{equ51}-\eqref{equ54} together, we obtain
\begin{equation}\label{equ28}
\begin{split}
I_{21}\leq & C+C\|(n^1,{\bf u}^1,\phi^1)\|^2_{H^3}\\
&+C(1+\varepsilon^3\|(n^1,{\bf u}^1)\|^3_{H^4})(\|(n^1,{\bf u}^1,\phi^1)\|^2_{H^s}+\varepsilon\|\nabla\phi^1\|^2_{H^s})\\
\leq & C(1+\varepsilon^3|||({\bf u}^1,\phi^1)|||^3_{\varepsilon,4}) (1+|||({\bf u}^1,\phi^1)|||^2_{\varepsilon,s}+|||({\bf u}^1,\phi^1)|||^2_{\varepsilon,3}),
\end{split}
\end{equation}
where we have used the definition of the norm $|||\cdot|||_{\varepsilon,s}$ in \eqref{trinorm} and Lemma \ref{lem2} to replace the norms of $n^1$ with the norms of $\phi^1$,
By putting \eqref{equ27} and \eqref{equ28} together and using Lemma \ref{lem2}, we obtain \eqref{equ5}. This ends the proof of Lemma \ref{lem6}.
\end{proof}

\begin{lemma}\label{lem7}
For $I_{1}$ in \eqref{equ6}, we have the estimates
\begin{equation*}
\begin{split}
I_{1}\leq &-\frac12\frac{d}{dt}\int\frac{n^0}{n^\varepsilon} |\partial^{\gamma}\phi^1|^2 -\frac{\varepsilon}{2}\frac{d}{dt}\int\frac{1}{n^\varepsilon} |\partial^{\gamma}\nabla\phi^1|^2\\
&+C(1+\varepsilon^3|||({\bf u}^1,\phi^1)|||^3_{\varepsilon,5})(1+|||({\bf u}^1,\phi^1)|||^2_{\varepsilon,s}+|||({\bf u}^1,\phi^1)|||^2_{\varepsilon,3}).
\end{split}
\end{equation*}
\end{lemma}
\begin{proof}
From the remainder equation \eqref{rem-p}, we obtain
\begin{equation}\label{equ18}
\begin{split}
\partial^{\gamma}\partial_tn^1=\partial^{\gamma}\partial_t(n^0\phi^1) -\varepsilon\partial^{\gamma}\partial_t\Delta\phi^1 -\partial^{\gamma}\partial_t\Delta\phi^0 -\sqrt{\varepsilon}\partial^{\gamma}\partial_tR^1.
\end{split}
\end{equation}
In this way, we can divide $I_{1}$ in \eqref{equ6} into the following
\begin{equation}\label{equ7}
\begin{split}
I_{1}=&\varepsilon\int\frac{\partial^{\gamma}\phi^1}{n^0+\varepsilon n^1}\partial^{\gamma}\partial_t\Delta\phi^1 -\int\frac{\partial^{\gamma}\phi^1}{n^0+\varepsilon n^1}\partial^{\gamma}\partial_t(n^0\phi^1) \\
&-\int\frac{\partial^{\gamma}\phi^1}{n^0+\varepsilon n^1}\partial^{\gamma}\partial_t\Delta\phi^0 -\sqrt{\varepsilon}\int\frac{\partial^{\gamma}\phi^1}{n^0+\varepsilon n^1}\partial^{\gamma}\partial_tR^1=\sum_{i=1}^4I_{1i}.
\end{split}
\end{equation}

In the following, we treat the RHS terms of \eqref{equ7} one by one.

\begin{itemize}
  \item \emph{Estimate of $I_{12}$.}
\end{itemize}
For the term $I_{12}$, we have
\begin{equation*}
\begin{split}
I_{12}=&-\int\frac{\partial^{\gamma}\phi^1}{n^0+\varepsilon n^1}\partial^{\gamma}(\phi^1\partial_tn^0) -\int\frac{\partial^{\gamma}\phi^1}{n^0+\varepsilon n^1}\partial^{\gamma}(n^0\partial_t\phi^1)\\
=&-\int\frac{\partial^{\gamma}\phi^1}{n^0+\varepsilon n^1}\partial^{\gamma}(\phi^1\partial_tn^0) -\int\frac{n^0\partial^{\gamma}\phi^1}{n^0+\varepsilon n^1}\partial_t\partial^{\gamma}\phi^1 -\int\frac{\partial^{\gamma}\phi^1}{n^0+\varepsilon n^1}[\partial^{\gamma},n^0]\partial_t\phi^1\\
=&:\sum_{i=1}^3I_{12i}.
\end{split}
\end{equation*}
For the first term $I_{121}$, since $n^0$ is known and is assumed to be smooth in Theorem \ref{thm1}, we have
\begin{equation*}
\begin{split}
I_{121}\leq C\|\phi^1\|^2_{H^s},
\end{split}
\end{equation*}
thanks to the multiplicative estimate in Lemma \ref{Le-inequ}. For the second term $I_{122}$, by integration by parts and Lemma \ref{lem3}, we have
\begin{equation*}
\begin{split}
I_{122}=&-\frac12\frac{d}{dt}\int\frac{n^0}{n^0+\varepsilon n^1}|\partial^{\gamma}\phi^1|^2 +\frac12\int\partial_t(\frac{n^0}{n^0+\varepsilon n^1})|\partial^{\gamma}\phi^1|^2\\
\leq & -\frac12\frac{d}{dt}\int\frac{n^0}{n^0+\varepsilon n^1}|\partial^{\gamma}\phi^1|^2 +C(1+\varepsilon\|\partial_tn^1\|_{L^{\infty}}) \|\partial^{\gamma}\phi^1\|^2_{L^2}\\
\leq & -\frac12\frac{d}{dt}\int\frac{n^0}{n^0+\varepsilon n^1}|\partial^{\gamma}\phi^1|^2 +C(1+\varepsilon^2|||({\bf u}^1,\phi^1)|||^2_{\varepsilon,3})\|\phi^1\|^2_{H^s},
\end{split}
\end{equation*}
where we have used the Sobolev embedding $H^2\hookrightarrow L^{\infty}$ for $d\leq 3$. For the third term $I_{123}$, by commutator estimates in Lemma \ref{Le-inequ} and Lemma \ref{lem3}, we have
\begin{equation*}
\begin{split}
I_{123}\leq & C\|\partial^{\gamma}\phi^1\|_{L^2}(\|\nabla n^0\|_{L^{\infty}}\|\partial_t\phi^1\|_{H^{s-1}} +\|n^0\|_{H^s}\|\partial_t\phi^1\|_{L^{\infty}})\\
\leq & C(1+|||({\bf u}^1,\phi^1)|||^2_{\varepsilon,s}+|||({\bf u}^1,\phi^1)|||^2_{\varepsilon,3}).
\end{split}
\end{equation*}
Summarizing, we have
\begin{equation}\label{equ22}
\begin{split}
I_{12}\leq &-\frac12\frac{d}{dt}\int\frac{n^0}{n^0+\varepsilon n^1}|\partial^{\gamma}\phi^1|^2\\
&+C(1+\varepsilon^2|||({\bf u}^1,\phi^1)|||^2_{\varepsilon,3})(1+|||({\bf u}^1,\phi^1)|||^2_{\varepsilon,s}+|||({\bf u}^1,\phi^1)|||^2_{\varepsilon,3}).
\end{split}
\end{equation}

\begin{itemize}
  \item \emph{Estimate of $I_{13}$.}
\end{itemize}
For the term $I_{13}$, since $\phi^0$ is known and smooth, it is easy to obtain
\begin{equation}\label{equ23}
\begin{split}
I_{13}\leq C(1+\|\phi^1\|^2_{H^s}).
\end{split}
\end{equation}

\begin{itemize}
  \item \emph{Estimate of $I_{14}$.}
\end{itemize}
For the term $I_{14}$, by integration by parts, we have
\begin{equation*}
\begin{split}
I_{14}=& \sqrt{\varepsilon}\int\frac{\partial^{\gamma+\gamma_1}\phi^1}{n^0+\varepsilon n^1}\partial^{\gamma-\gamma_1}\partial_tR^1 +\sqrt{\varepsilon}\int\partial^{\gamma_1}(\frac{1}{n^0+\varepsilon n^1})\partial^{\gamma}\phi^1\partial^{\gamma-\gamma_1}\partial_tR^1,
\end{split}
\end{equation*}
where $\gamma_1\leq\gamma$ is a multiindex with $|\gamma_1|=1$. By Lemma \ref{lem2}, Corollary \ref{cor1} and the definition of the triple norm $|||\cdot|||_{\varepsilon,s}$, we have the bound
\begin{equation}\label{equ24}
\begin{split}
I_{14}\leq & C\sqrt\varepsilon\|\partial^{\gamma+\gamma_1}\phi^1\|_{L^2} \|\partial_tR^1\|_{H^{s-1}} +C(1+\varepsilon\|\partial^{\gamma_1}n^1\|_{L^{\infty}}) \|\partial^{\gamma}\phi^1\|_{L^2}\|\partial_tR^1\|_{H^{s-1}}\\
\leq & C\|\partial_tR^1\|^2_{H^{s-1}}+C\varepsilon\|\nabla\phi^1\|_{H^s}^2 +C(1+\varepsilon^2\|n^1\|^2_{H^3})\|\phi^1\|_{H^s}^2\\
\leq & C(1+\varepsilon^2|||({\bf u}^1,\phi^1)|||^2_{\varepsilon,3})(1+|||({\bf u}^1,\phi^1)|||^2_{\varepsilon,s}).
\end{split}
\end{equation}

\begin{itemize}
  \item \emph{Estimate of $I_{11}$.}
\end{itemize}
We next deal with the term $I_{11}$ in \eqref{equ7}. By integration by parts, we have
\begin{equation*}
\begin{split}
I_{11}=&-\varepsilon\int\frac{1}{n^0+\varepsilon n^1} \partial^{\gamma}\nabla\phi^1\cdot\partial_t\partial^{\gamma}\nabla\phi^1 -\varepsilon\int\nabla(\frac{1}{n^0+\varepsilon n^1}) \partial^{\gamma}\phi^1\cdot\partial_t\partial^{\gamma}\nabla\phi^1\\
=&-\frac{\varepsilon}{2}\frac{d}{dt}\int\frac{1}{n^0+\varepsilon n^1} |\partial^{\gamma}\nabla\phi^1|^2 +\frac{\varepsilon}{2}\int\partial_t(\frac{1}{n^0+\varepsilon n^1}) |\partial^{\gamma}\nabla\phi^1|^2 \\
&-\varepsilon\int\nabla(\frac{1}{n^0+\varepsilon n^1}) \partial^{\gamma}\phi^1\cdot\partial_t\partial^{\gamma}\nabla\phi^1 =:I_{111}+I_{112}+I_{113}.
\end{split}
\end{equation*}
Since from Sobolev embedding and Lemma \ref{lem3}, we have
\begin{equation*}
\begin{split}
\|\partial_t(\frac{n^0}{n^0+\varepsilon n^1})\|_{L^{\infty}}\leq C(1+\varepsilon\|\partial_tn^1\|_{L^{\infty}})
\leq  C(1+\varepsilon^2|||({\bf u}^1, \phi^1)|||^2_{\varepsilon,3}),
\end{split}
\end{equation*}
it is immediate that
\begin{equation*}
\begin{split}
 I_{112}\leq C(1+\varepsilon^2|||({\bf u}^1, \phi^1)|||^2_{\varepsilon,3}) (\varepsilon\|\partial^{\gamma}\nabla\phi^1\|^2).
\end{split}
\end{equation*}
For the term $I_{113}$, we have by integration by parts,
\begin{equation}\label{equ8}
\begin{split}
I_{113}=&\varepsilon\int\Delta(\frac{1}{n^0+\varepsilon n^1}) \partial^{\gamma}\phi^1\partial_t\partial^{\gamma}\phi^1 +\varepsilon\int\nabla(\frac{1}{n^0+\varepsilon n^1}) \partial^{\gamma}\nabla\phi^1\partial_t\partial^{\gamma}\phi^1\\
=&-\varepsilon\int\partial^{\gamma_1}\Delta(\frac{1}{n^\varepsilon}) \partial^{\gamma}\phi^1\partial_t\partial^{\gamma-\gamma_1}\phi^1 -\varepsilon\int\Delta(\frac{1}{n^\varepsilon}) \partial^{\gamma+\gamma_1}\phi^1\partial_t\partial^{\gamma-\gamma_1}\phi^1 \\
&-\varepsilon\int\partial^{\gamma_1}\nabla(\frac{1}{n^\varepsilon}) \nabla\partial^{\gamma}\phi^1\partial_t\partial^{\gamma-\gamma_1}\phi^1 -\varepsilon\int\nabla(\frac{1}{n^\varepsilon}) \nabla\partial^{\gamma+\gamma_1}\phi^1\partial_t\partial^{\gamma-\gamma_1}\phi^1, \end{split}
\end{equation}
where $\gamma_1\leq\gamma$ is a multiindex with $|\gamma_1|=1$. By direct computation, H\"older inequality and Sobolev embedding $H^2\hookrightarrow L^{\infty}$, it is easy to obtain
\begin{equation}\label{equ9}
\begin{split}
\|\partial^{\alpha}(\frac{1}{n^\varepsilon})\|_{L^{\infty}}\leq & C(1+\varepsilon^{|\alpha|}\|n^1\|^{|\alpha|}_{H^{2+|\alpha|}})\\
\leq & C(1+\varepsilon^{|\alpha|}|||({\bf u}^1,\phi^1)|||^{|\alpha|}_{\varepsilon,{2+|\alpha|}})
\end{split}
\end{equation}
for any smooth function $n^0$ and any multiindex $\alpha$, thanks to Lemma \ref{lem2}. On the other hand, from Lemma \ref{lem4} and Lemma \ref{lem3}, we obtain
\begin{equation}\label{equ10}
\begin{split}
\|\partial_t\partial^{\gamma-\gamma_1}\phi^1\|^2_{L^{2}}\leq & C(1+\|\phi^1\|^2_{H^{s-1}}+\|\partial_tn^1\|^2_{H^{s-1}})\\
\leq & C(1+|||({\bf u}^1,\phi^1)|||^2_{\varepsilon,s}).
\end{split}
\end{equation}
Since the order of the derivatives on $n^0/(n^0+\varepsilon n^1)$ in \eqref{equ8} does not exceed 3, by using H\"older inequality, \eqref{equ9} and \eqref{equ10}, $I_{113}$ can be bounded by
\begin{equation*}
\begin{split}
I_{113}\leq & C(1+\varepsilon^{3}|||({\bf u}^1,\phi^1)|||^{3}_{\varepsilon,5}) \cdot\varepsilon(\|\partial^{\gamma}\phi^1\|_{L^2} +\|\partial^{\gamma+\gamma_1}\phi^1\|_{L^2}\\ &+\|\nabla\partial^{\gamma}\phi^1\|_{L^2} +\|\nabla\partial^{\gamma+\gamma_1}\phi^1\|_{L^2})\cdot \|\partial_t\partial^{\gamma-\gamma_1}\phi^1\|_{L^2}\\
\leq & C(1+\varepsilon^{3}|||({\bf u}^1,\phi^1)|||^{3}_{\varepsilon,5}) \cdot\\
&\cdot(\|\partial^{\gamma}\phi^1\|^2_{L^2} +\varepsilon\|\partial^{\gamma}\nabla\phi^1\|^2_{L^2} +\varepsilon^2\|\partial^{\gamma}\Delta\phi^1\|^2_{L^2} +\|\partial_t\partial^{\gamma-\gamma_1}\phi^1\|^2_{L^2})\\
\leq & C(1+\varepsilon^{3}|||({\bf u}^1,\phi^1)|||^{3}_{\varepsilon,5})(1+|||({\bf u}^1,\phi^1)|||^2_{\varepsilon,s}),
\end{split}
\end{equation*}
where we have used the definition of $|||\cdot|||_{\varepsilon,s}$ in \eqref{trinorm} and the $L^2$ boundedness of the Riesz operator \cite{Stein70}. To be more precise, there exists some constant $C>0$ such that $\|\partial_i\partial_j\phi^1\|\leq C\|\Delta\phi^1\|$
since $\partial_i\partial_j=-R_iR_j\Delta$, where $R_i$ is the $i^{th}$ Riesz operator. In particular, we have $\|\nabla\partial^{\gamma+\gamma_1}\phi^1\|\leq C\|\partial^{\gamma}\Delta\phi\|$.

Summarizing, we have
\begin{equation}\label{equ25}
\begin{split}
I_{11}=&-\frac{\varepsilon}{2}\frac{d}{dt}\int\frac{1}{n^\varepsilon} |\partial^{\gamma}\nabla\phi^1|^2 +C(1+\varepsilon^{3}|||({\bf u}^1,\phi^1)|||^{3}_{\varepsilon,5})(1+|||({\bf u}^1,\phi^1)|||^2_{\varepsilon,s}).
\end{split}
\end{equation}

By \eqref{equ7} and the estimates of \eqref{equ22}, \eqref{equ23}, \eqref{equ24} and \eqref{equ25}, we have
\begin{equation*}
\begin{split}
I_{1}\leq &-\frac12\frac{d}{dt}\int\frac{n^0}{n^\varepsilon} |\partial^{\gamma}\phi^1|^2 -\frac{\varepsilon}{2}\frac{d}{dt}\int\frac{1}{n^\varepsilon} |\partial^{\gamma}\nabla\phi^1|^2\\
&+C(1+\varepsilon^3|||({\bf u}^1,\phi^1)|||^3_{\varepsilon,5})(1+|||({\bf u}^1,\phi^1)|||^2_{\varepsilon,s}+|||({\bf u}^1,\phi^1)|||^2_{\varepsilon,3}).
\end{split}
\end{equation*}
This ends the proof of Lemma \ref{lem7}.
\end{proof}

Now, we can end the proof of Lemma \ref{lem5}.

\begin{proof}[End of proof of Lemma \ref{lem5}]The proof of Lemma \ref{lem5} is closed by \eqref{equ4}, \eqref{equ30}, \eqref{equ31}, \eqref{equ26}, \eqref{equ29} and Lemma \ref{lem6} and \ref{lem7}.
\end{proof}

\begin{proposition}\label{prop1}
Let $s\geq0$ be a positive integer, $(n^1,{\bf u}^1,\phi^1)$ be a smooth solution for the system \eqref{rem}. There exists $\varepsilon_1>0$ and $C,C'>0$ such that for any $0<\varepsilon<\varepsilon_1$, there holds
\begin{equation}\label{equ11}
\begin{split}
\|{\bf u}^1&(t)\|^2_{H^s}+\|\phi^1(t)\|^2_{H^s} +\varepsilon\|\nabla\phi^1(t)\|^2_{H^s}\\
\leq & C'(\|{\bf u}^1(0)\|^2_{H^s}+\|\phi^1(0)\|^2_{H^s} +\varepsilon\|\nabla\phi^1(0)\|^2_{H^s})\\
&+C\int_0^t(1+\varepsilon^3|||({\bf u}^1,\phi^1)|||^3_{\varepsilon,5})(1+|||({\bf u}^1,\phi^1)|||^2_{\varepsilon,s}+|||({\bf u}^1,\phi^1)|||^2_{\varepsilon,3})d\tau.
\end{split}
\end{equation}
where $C'$ depends only on $\sigma'$ and $\sigma''$.
\end{proposition}
\begin{proof}
This is shown by integrating \eqref{equ101} over $[0,t]$ and summing them up for $|\gamma|\leq s$, and then using $\sigma'<n^0<\sigma''$ and $\frac{\sigma'}{2}<n^{\varepsilon}<2\sigma''$ for any $t\in [0,T_{\varepsilon}]$ in \eqref{blowup} and \eqref{assumption1} for $0<\varepsilon<\varepsilon_1$.
\end{proof}

However, this Gronwall inequality is not closed since the right hand side of \eqref{equ11} depends on $\varepsilon^2\|\Delta\phi^1\|_{H^s}^2$, which does not appear on the left hand side. This will be treated in the next subsection.

\subsection{Weighted $s+1$ order estimates}
In the following, we also let $\gamma$ be a multiindex with $|\gamma|=s$. The main result in this subsection is Proposition \ref{prop2}.

\begin{lemma}\label{lem8}
Let $s\geq0$ be a positive integer, $(n^1,{\bf u}^1,\phi^1)$ be a smooth solution for the system \eqref{rem}. There exists $\varepsilon_1>0$ and $C,C'>0$ such that
\begin{equation}\label{equ45}
\begin{split}
\frac{\varepsilon}{2}\frac{d}{dt}\int|\partial^{\gamma}&\nabla{\bf u}^1|^2 \leq -\frac{\varepsilon}{2}\frac{d}{dt}\int\frac{ n^0}{n^{\varepsilon}}|\partial^{\gamma}\nabla\phi^1|^2 -\frac{\varepsilon^2}{2}\frac{d}{dt}\int\frac{1}{n^\varepsilon}|\partial^{\gamma}\Delta\phi^1|^2\\
&+C(1+\varepsilon^2|||({\bf u}^1,\phi^1)|||_{\varepsilon,4}^2)(1+|||({\bf u}^1,\phi^1)|||_{\varepsilon,s}^2+|||({\bf u}^1,\phi^1)|||_{\varepsilon,3}^2),
\end{split}
\end{equation}
for any $0<\varepsilon<\varepsilon_1$.
\end{lemma}
\begin{proof}
Let $\gamma$ be a multiindex with $|\gamma|=s\geq0$. Taking $\partial^{\gamma}$ in the second equation of \eqref{rem}, we obtain
\begin{equation*}
\partial_t\partial^{\gamma}{\bf u}^{1}+\partial^{\gamma}({\bf u}^0\cdot\nabla {\bf u}^1)+\partial^{\gamma}({\bf u}^1\cdot\nabla {\bf u}^0)+\varepsilon \partial^{\gamma}({\bf u}^1\cdot\nabla {\bf u}^1)=-\partial^{\gamma}\nabla\phi^1.
\end{equation*}
Taking $L^2$ inner product with $-\varepsilon \partial^{\gamma}\Delta{\bf u}^1$, we obtain
\begin{equation}\label{equ12}
\begin{split}
\frac{\varepsilon}{2}\frac{d}{dt}\int|\partial^{\gamma}\nabla{\bf u}^1|^2 =&-\varepsilon\int\partial_t\partial^{\gamma}{\bf u}^1\partial^{\gamma}\Delta{\bf u}^1\\
=&\varepsilon\int\partial^{\gamma}\nabla\phi^1\partial^{\gamma}\Delta{\bf u}^1 +\varepsilon^2\int\partial^{\gamma}({\bf u}^1\cdot\nabla {\bf u}^1)\partial^{\gamma}\Delta{\bf u}^1\\
&+\varepsilon\int\partial^{\gamma}({\bf u}^0\cdot\nabla {\bf u}^1)\partial^{\gamma}\Delta{\bf u}^1 +\varepsilon\int\partial^{\gamma}({\bf u}^1\cdot\nabla {\bf u}^0)\partial^{\gamma}\Delta{\bf u}^1\\
=&:I^{\varepsilon}+II^{\varepsilon}+III^{\varepsilon}+IV^{\varepsilon}.
\end{split}
\end{equation}

\begin{itemize}
  \item \emph{Estimate of $IV^{\varepsilon}$.}
\end{itemize}
The term $IV^{\varepsilon}$ can be bounded by
\begin{equation}\label{equ41}
\begin{split}
IV^{\varepsilon}=&-\varepsilon\int\partial^{\gamma}\nabla({\bf u}^1\cdot\nabla {\bf u}^0)\partial^{\gamma}\nabla{\bf u}^1\\
=&-\varepsilon\int\left(\partial^{\gamma}(\nabla{\bf u}^1\cdot\nabla {\bf u}^0) +\partial^{\gamma}({\bf u}^1\cdot\nabla^2{\bf u}^0)\right)\partial^{\gamma}\nabla{\bf u}^1\\
\leq & C\varepsilon(\|\partial^{\gamma}(\nabla{\bf u}^1\cdot\nabla{\bf u}^0)\|^2_{L^2}+\|\partial^{\gamma}({\bf u}^1\cdot\nabla^2{\bf u}^0)\|^2_{L^2})+C\varepsilon\|\partial^{\gamma}\nabla{\bf u}^1\|^2_{L^2}\\
\leq & C(\|{\bf u}^1\|^2_{H^s}+\varepsilon\|\nabla{\bf u}^1\|^2_{H^s}+\|{\bf u}^1\|^2_{H^3}),
\end{split}
\end{equation}
where we have used the commutator estimates \eqref{commutator}, the Sobolev embedding $H^2\hookrightarrow L^{\infty}$ when $d\leq 3$ and the fact that $n^0$ and ${\bf u}^0$ are known smooth solutions of the Euler equation \eqref{EQ} by Theorem \ref{thm1}.

\begin{itemize}
  \item \emph{Estimate of $III^{\varepsilon}$.}
\end{itemize}
By integration by parts, the third term $III^{\varepsilon}$ can be rewritten as
\begin{equation*}
\begin{split}
III^{\varepsilon}=&\varepsilon\int{\bf u}^0\cdot\nabla\partial^{\gamma}{\bf u}^1\partial^{\gamma}\Delta{\bf u}^1 +\varepsilon\int[\partial^{\gamma},{\bf u}^0]\cdot\nabla{\bf u}^1\partial^{\gamma}\Delta{\bf u}^1\\
=&\varepsilon\int\sum_{i,j,k}{u}^0_i\partial_i\partial^{\gamma}{u}^1_k \partial^{\gamma}\partial_j\partial_j{u}^1_k +\varepsilon\int\sum_{\beta=1}^{\gamma}C_{\gamma}^{\beta} \partial^{\beta}{\bf u}^0\cdot\partial^{\gamma-\beta}\nabla{\bf u}^1\cdot \partial^{\gamma}\Delta{\bf u}^1\\
=&-2\varepsilon\int\sum_{i,j,k}\partial_j{u}^0_i \partial_i\partial^{\gamma}{u}^1_k \partial^{\gamma}\partial_j{u}^1_k +\varepsilon\int\sum_{i,j,k}\partial_i{u}^0_i \partial^{\gamma}\partial_j{u}^1_k \partial^{\gamma}\partial_j{u}^1_k\\
&-\varepsilon\int\sum_{\beta=1}^{\gamma}C_{\gamma}^{\beta} \partial^{\beta}{\bf u}^0\cdot\partial^{\gamma-\beta}\nabla^2{\bf u}^1\cdot \partial^{\gamma}\nabla{\bf u}^1 -\varepsilon\int\sum_{\beta=1}^{\gamma}C_{\gamma}^{\beta} \partial^{\beta}\nabla{\bf u}^0\cdot\partial^{\gamma-\beta}\nabla{\bf u}^1\cdot \partial^{\gamma}\nabla{\bf u}^1\\
=&:III^{\varepsilon}_1+III^{\varepsilon}_2+III^{\varepsilon}_3 +III^{\varepsilon}_4,
\end{split}
\end{equation*}
where ${\bf u}^0=(u^0_1,\cdots,u^0_d)$ and ${\bf u}^1=(u^1_1,\cdots,u^1_d)$ and we have used integration by parts twice in the third equality. For $III^{\varepsilon}_1$, $III^{\varepsilon}_2$ and $III^{\varepsilon}_3$, we easily obtain
\begin{equation*}
\begin{split}
|III^{\varepsilon}_1,III^{\varepsilon}_2,III^{\varepsilon}_3|\leq C\varepsilon\|\nabla{\bf u}^1\|_{H^s}^2.
\end{split}
\end{equation*}
For $III^{\varepsilon}_4$, by H\"older inequality, we obtain
\begin{equation*}
\begin{split}
|III^{\varepsilon}_4|\leq C\varepsilon\|{\bf u}^1\|_{H^s}\|\nabla{\bf u}^1\|_{H^s}\leq C(\|{\bf u}^1\|_{H^s}^2+\varepsilon\|\nabla{\bf u}^1\|_{H^s}^2).
\end{split}
\end{equation*}
Summing them up, we obtain
\begin{equation}\label{equ42}
\begin{split}
|III^{\varepsilon}|\leq C(\|{\bf u}^1\|_{H^s}^2+\varepsilon\|\nabla{\bf u}^1\|_{H^s}^2).
\end{split}
\end{equation}

\begin{itemize}
  \item \emph{Estimate of $II^{\varepsilon}$.}
\end{itemize}
For the second term $II^{\varepsilon}$, by integration by parts, we have
\begin{equation*}
\begin{split}
II^{\varepsilon}=&\varepsilon^2\int{\bf u}^1\cdot\partial^{\gamma}\nabla{\bf u}^1 \partial^{\gamma}\Delta{\bf u}^1 +\varepsilon^2\int[\partial^{\gamma},{\bf u}^1]\cdot\nabla{\bf u}^1\partial^{\gamma}\Delta{\bf u}^1\\
=&-\varepsilon^2\int\sum_{i,j,k}\partial_j{u}^1_i \partial^{\gamma}\partial_i{u}^1_k\partial^{\gamma}\partial_j{u}^1_k +\frac{\varepsilon^2}{2}\int\sum_{i,j,k}\partial_i{u}^1_i \partial^{\gamma}\partial_j{u}^1_k\partial^{\gamma}\partial_j{u}^1_k\\
&-\varepsilon^2\int[\partial^{\gamma},\nabla{\bf u}^1]\cdot\nabla{\bf u}^1\cdot\partial^{\gamma}\nabla{\bf u}^1 -\varepsilon^2\int[\partial^{\gamma},{\bf u}^1]\cdot\nabla^2{\bf u}^1\cdot\partial^{\gamma}\nabla{\bf u}^1 =:\sum_{i=1}^4II^{\varepsilon}_{i}.
\end{split}
\end{equation*}
By H\"older inequality, $II^{\varepsilon}_{1}$ and $II^{\varepsilon}_{2}$ are bounded by
\begin{equation*}
\begin{split}
|II^{\varepsilon}_{1},II^{\varepsilon}_{2}|\leq & C(\varepsilon\|\nabla{\bf u}^1\|_{L^{\infty}})(\varepsilon\|\nabla{\bf u}^1\|_{H^s}^2)\\
\leq & C(\varepsilon\|{\bf u}^1\|_{H^3})(\varepsilon\|\nabla{\bf u}^1\|_{H^s}^2).
\end{split}
\end{equation*}
By commutator estimates,
\begin{equation*}
\begin{split}
\|[\partial^{\gamma},\nabla{\bf u}^1]\cdot\nabla{\bf u}^1\|_{L^2}\leq & C(\|\nabla^2{\bf u}^1\|_{L^{\infty}}\|\nabla{\bf u}^1\|_{H^{s-1}} +\|\nabla{\bf u}^1\|_{H^s}\|\nabla{\bf u}^1\|_{L^{\infty}})\\
\leq & C\|{\bf u}^1\|_{H^4}(\|{\bf u}^1\|_{H^{s}}+\|\nabla{\bf u}^1\|_{H^s}).
\end{split}
\end{equation*}
Hence
\begin{equation*}
\begin{split}
|II^{\varepsilon}_{3}|\leq  C\varepsilon\|{\bf u}^1\|_{H^4}(\|{\bf u}^1\|^2_{H^{s}}+\varepsilon\|\nabla{\bf u}^1\|^2_{H^s}).
\end{split}
\end{equation*}
Similarly, by commutator estimates
\begin{equation*}
\begin{split}
|II^{\varepsilon}_{4}|\leq & C\varepsilon^2(\|\nabla{\bf u}^1\|_{L^{\infty}}\|\nabla^2{\bf u}^1\|_{H^{s-1}} +\|{\bf u}^1\|_{H^s}\|\nabla^2{\bf u}^1\|_{L^{\infty}})\|\nabla{\bf u}^1\|_{H^s}\\
\leq & C\varepsilon\|{\bf u}^1\|_{H^4}(\|{\bf u}^1\|^2_{H^{s}}+\varepsilon\|\nabla{\bf u}^1\|^2_{H^s}).
\end{split}
\end{equation*}
Summing them up, we obtain
\begin{equation}\label{equ43}
\begin{split}
|II^{\varepsilon}|\leq & C\varepsilon\|{\bf u}^1\|_{H^4}(\|{\bf u}^1\|^2_{H^{s}}+\varepsilon\|\nabla{\bf u}^1\|^2_{H^s}).
\end{split}
\end{equation}

\begin{itemize}
  \item \emph{Estimate of $I^{\varepsilon}$.}
\end{itemize}
In the following, we treat the first term $I^{\varepsilon}$ in \eqref{equ12}. By integration by parts thrice, we obtain
\begin{equation*}
\begin{split}
I^{\varepsilon}=-\varepsilon\int\partial^{\gamma}\Delta\phi^1 \partial^{\gamma}\nabla\cdot{\bf u}^1.
\end{split}
\end{equation*}
By using \eqref{equ13}, we obtain
\begin{equation}\label{equ14}
\begin{split}
I^{\varepsilon}=& \int\frac{\varepsilon\partial^{\gamma}\Delta\phi^1}{n^\varepsilon} \partial_t\partial^{\gamma}n^1 +\int\frac{\varepsilon\partial^{\gamma}\Delta\phi^1}{n^\varepsilon} \partial^{\gamma}({\bf u}^\varepsilon\cdot\nabla n^1)\\
&+\int\frac{\varepsilon\partial^{\gamma}\Delta\phi^1}{n^\varepsilon} [\partial^{\gamma},n^\varepsilon]\nabla\cdot{\bf u}^1
+\int\frac{\varepsilon\partial^{\gamma}\Delta\phi^1}{n^\varepsilon} \partial^{\gamma}({\bf u}^1\cdot\nabla n^0)\\
&+\int\frac{\varepsilon\partial^{\gamma}\Delta\phi^1}{n^\varepsilon} \partial^{\gamma}({n}^1\cdot\nabla {\bf u}^0) =:\sum_{i=1}^5I^{\varepsilon}_{i}.
\end{split}
\end{equation}

In the following, we estimate $I^{\varepsilon}_{3}\sim I^{\varepsilon}_{5}$ while leaving the estimates of $I^{\varepsilon}_{1}$ and $I^{\varepsilon}_{2}$ to the next two lemmas. For $I^{\varepsilon}_{3}$, we have
\begin{equation*}
\begin{split}
I^{\varepsilon}_{3}\leq & C\varepsilon\|\partial^{\gamma}\Delta\phi^1\|_{L^2}(\|\nabla(n^0+\varepsilon n^1)\|_{L^{\infty}}\|\nabla\cdot{\bf u}^1\|_{H^{s-1}}+\|n^0+\varepsilon n^1\|_{H^{s}}\|\nabla\cdot{\bf u}^1\|_{L^{\infty}})\\
\leq & C\varepsilon^2\|\Delta\phi^1\|^2_{H^s}+ \|{\bf u}^1\|^2_{H^s}+\|{\bf u}^1\|^2_{H^3} +C\varepsilon^2(\|n^1\|^2_{H^3}+\|{\bf u}^1\|^2_{H^3})(\|n^1\|^2_{H^s}+\|{\bf u}^1\|^2_{H^s}),
\end{split}
\end{equation*}
where we have used the commutator estimates \eqref{commutator} and the fact that $n^0+\varepsilon n^1$ are bounded from above and below by positive numbers when $\varepsilon<\varepsilon_1$ is small enough in \eqref{assumption1}. Recalling Lemma \ref{lem2} and the definition of the triple norm \eqref{trinorm}, we obtain
\begin{equation*}
\begin{split}
I^{\varepsilon}_{3}\leq & C(1+\varepsilon^2|||({\bf u}^1,\phi)|||^2_{\varepsilon,3})|||({\bf u}^1,\phi)|||^2_{\varepsilon,s}+C\|{\bf u}^1\|^2_{H^3}.
\end{split}
\end{equation*}
For $I^{\varepsilon}_{4}$, we have
\begin{equation*}
\begin{split}
I^{\varepsilon}_{4}\leq & C\varepsilon\|\Delta\phi^1\|_{H^s}(\|{\bf u}^1\|_{L^\infty}\|\nabla n^0\|_{H^s}+\|{\bf u}^1\|_{H^s}\|\nabla n^0\|_{L^\infty})\\
\leq & C\varepsilon^2\|\Delta\phi^1\|^2_{H^s}+C\|{\bf u}^1\|^2_{H^s}+C\|{\bf u}^1\|^2_{H^2}.
\end{split}
\end{equation*}
Similarly, $I^{\varepsilon}_{5}$ can be bounded by
\begin{equation*}
\begin{split}
I^{\varepsilon}_{5}\leq & C\varepsilon\|\Delta\phi^1\|^2_{H^s} +C\|n^1\|^2_{H^s}+C\|n^1\|^2_{H^2}\\
\leq & C(1+|||\phi|||^2_{\varepsilon,s})+C\|n^1\|^2_{H^2},
\end{split}
\end{equation*}
thanks to Lemma \ref{lem2}. Summarizing, we obtain from \eqref{equ14}
\begin{equation}\label{equ44}
\begin{split}
I^{\varepsilon}\leq & I^{\varepsilon}_{1}+I^{\varepsilon}_{2} +C\|{\bf u}^1\|^2_{H^3}+C\|n^1\|^2_{H^2}\\
&+C(1+\varepsilon^2|||({\bf u}^1,\phi)|||^2_{\varepsilon,3})(1+|||({\bf u}^1,\phi)|||^2_{\varepsilon,s}),
\end{split}
\end{equation}
where $I^{\varepsilon}_{1}$ and $I^{\varepsilon}_{2}$ will be treated in the next two lemmas.
\end{proof}

\begin{lemma}\label{lem9}
The term $I^{\varepsilon}_2$ in \eqref{equ14} can be estimated as
\begin{equation*}
\begin{split}
I^{\varepsilon}_{2}\leq C(1+\varepsilon^2|||({\bf u}^1,\phi^1)|||_{\varepsilon,3}^2)(1+|||({\bf u}^1,\phi^1)|||_{\varepsilon,s}^2)+|||\phi^1|||_{\varepsilon,3}^2,
\end{split}
\end{equation*}
for all $0<\varepsilon<\varepsilon_1$ for some $\varepsilon_1>0$.
\end{lemma}
\begin{proof}
Recall that $I^{\varepsilon}_{2}$ is given by \eqref{equ14}.
By integration by parts, it can be rewritten as
\begin{equation*}
\begin{split}
I^{\varepsilon}_{2}=\int\frac{\varepsilon\partial^{\gamma}\Delta\phi^1} {n^\varepsilon}{\bf u}^\varepsilon\cdot\nabla\partial^{\gamma} n^1 +\int\frac{\varepsilon\partial^{\gamma}\Delta\phi^1}{n^\varepsilon} [\partial^{\gamma},{\bf u}^\varepsilon]\cdot\nabla n^1=:I^{\varepsilon}_{21}+I^{\varepsilon}_{22}.
\end{split}
\end{equation*}

\begin{itemize}
  \item \emph{Estimate of $I^{\varepsilon}_{22}$.}
\end{itemize}
For the second term $I^{\varepsilon}_{22}$, we obtain
\begin{equation}\label{equ32}
\begin{split}
I^{\varepsilon}_{22}\leq & C\varepsilon\|\partial^{\gamma}\Delta\phi^1\|_{L^2}\|[\partial^{\gamma}, {\bf u}^\varepsilon]\cdot\nabla n^1\|_{L^2}\\
\leq & C\varepsilon\|\partial^{\gamma}\Delta\phi^1\|_{L^2}(\|\nabla{\bf u}^\varepsilon\|_{L^{\infty}}\|\nabla n^1\|_{H^{s-1}}+\|{\bf u}^\varepsilon\|_{H^s}\|\nabla n^1\|_{L^{\infty}})\\
\leq & C\varepsilon\|\Delta\phi^1\|_{H^s}(\|n^1\|_{H^s}+\|n^1\|_{H^3} +\varepsilon\|{\bf u}^1\|_{H^3}\|n^1\|_{H^s} +\varepsilon\|n^1\|_{H^3}\|{\bf u}^1\|_{H^s})\\
\leq & C\varepsilon^2\|\Delta\phi^1\|^2_{H^s} +C(\|n^1\|^2_{H^s}+\|n^1\|^2_{H^3} +\varepsilon^2\|({\bf u}^1,n^1)\|^2_{H^3}(\|(n^1,{\bf u}^1)\|^2_{H^s}))\\
\leq & C(1+\varepsilon^2|||({\bf u}^1,\phi^1)|||_{\varepsilon,3}^2)(1+|||({\bf u}^1,\phi^1)|||_{\varepsilon,s}^2)+|||\phi^1|||_{\varepsilon,3}^2,
\end{split}
\end{equation}
thanks to the commutator estimates \eqref{commutator} in the second step, Lemma \ref{lem2} in the last step and the definition of the triple norm in \eqref{trinorm}.

\begin{itemize}
  \item \emph{Estimate of $I^{\varepsilon}_{21}$.}
\end{itemize}
In the rest of this lemma, we focus us on the treatment of $I^{\varepsilon}_{21}$. Recalling the remainder equation \eqref{rem-p} (see also \eqref{equ15}), $I^{\varepsilon}_{21}$ can be divided into
\begin{equation}\label{equ17}
\begin{split}
I^{\varepsilon}_{21}=&\int\frac{\varepsilon \partial^{\gamma}\Delta\phi^1}{n^\varepsilon}{\bf u}^\varepsilon\cdot\nabla\partial^{\gamma}(n^0\phi^1) -\varepsilon^2\int\frac{\partial^{\gamma}\Delta\phi^1}{n^\varepsilon} {\bf u}^\varepsilon\cdot\partial^{\gamma}\nabla\Delta\phi^1\\
&-\int\frac{\varepsilon\partial^{\gamma}\Delta\phi^1}{n^\varepsilon}{\bf u}^\varepsilon\cdot\partial^{\gamma}\nabla\Delta\phi^0 -\sqrt\varepsilon\int\frac{\varepsilon\partial^{\gamma}\Delta\phi^1} {n^\varepsilon}{\bf u}^\varepsilon\cdot\partial^{\gamma}\nabla R^1=:\sum_{i=1}^4I^{\varepsilon}_{21i}.
\end{split}
\end{equation}

In the following, we estimates the four terms on the RHS of \eqref{equ17}. For the first term $I^{\varepsilon}_{211}$, we have
\begin{equation}\label{equ16}
\begin{split}
I^{\varepsilon}_{211}=& \int\frac{\varepsilon n^0\partial^{\gamma}\Delta\phi^1}{n^\varepsilon}{\bf u}^\varepsilon\cdot \nabla\partial^{\gamma}\phi^1 +\int\frac{\varepsilon\partial^{\gamma}\Delta\phi^1}{n^\varepsilon}{\bf u}^\varepsilon\cdot[\partial^{\gamma},n^0]\nabla\phi^1\\
=&\frac{\varepsilon}2\int\nabla\cdot\left(\frac{n^0{\bf u}^\varepsilon}{n^\varepsilon}\right)|\partial^{\gamma}\nabla\phi^1|^2 -\varepsilon\int\sum_{i,j}\partial_i \left(\frac{n^0{u}_j^\varepsilon}{n^\varepsilon}\right) \partial^{\gamma}\partial_i\phi^1\partial^{\gamma}\partial_j\phi^1\\
&+\int\frac{\varepsilon\partial^{\gamma}\Delta\phi^1}{n^\varepsilon} {\bf u}^\varepsilon\cdot[\partial^{\gamma},n^0]\nabla\phi^1 =:I^{\varepsilon}_{2111}+I^{\varepsilon}_{2112}+I^{\varepsilon}_{2113},
\end{split}
\end{equation}
where $u_j^{\varepsilon}=({\bf u}^{\varepsilon})_j=u_j^0+{\varepsilon}u_j^1$. By using Sobolev embedding $H^2\hookrightarrow L^{\infty}$,
\begin{equation*}
\begin{split}
\|\nabla\left(\frac{n^0{\bf u}^\varepsilon}{n^\varepsilon}\right) \|_{L^{\infty}}\leq C+C\varepsilon^2(\|n^1\|^2_{H^3}+\|{\bf u}^1\|^2_{H^3}),
\end{split}
\end{equation*}
which yields for the first two terms on the RHS of \eqref{equ16}
\begin{equation*}
\begin{split}
|I^{\varepsilon}_{2111},I^{\varepsilon}_{2112}|\leq & C\varepsilon(1+\varepsilon^2(\|n^1\|^2_{H^3}+\|{\bf u}^1\|^2_{H^3}))\|\nabla\phi^1\|^2_{H^s}\\
\leq & C(1+\varepsilon^2|||({\bf u}^1,\phi^1)|||^2_{\varepsilon,3})|||\phi^1|||^2_{\varepsilon,s}.
\end{split}
\end{equation*}
For the term $I^{\varepsilon}_{2113}$, by the commutator estimates \eqref{commutator}, we have
\begin{equation*}
\begin{split}
\|[\partial^{\gamma},n^0]\nabla\phi^1\|_{L^2}\leq & C(\|\nabla n^0\|_{L^{\infty}}\|\nabla\phi^1\|_{H^{s-1}}+\|n^0\|_{H^s}\|\nabla\phi^1\|_{L^{\infty}})\\
\leq & C(\|\phi^1\|_{H^3}+\|\phi^1\|_{H^s}),
\end{split}
\end{equation*}
which implies that by H\"older inequality
\begin{equation*}
\begin{split}
I^{\varepsilon}_{2113}\leq & C\|\frac{1}{n^0+\varepsilon n^1}\|_{L^{\infty}}\|\varepsilon\partial^{\gamma}\Delta\phi^1({\bf u}^0+\varepsilon{\bf u}^1)\|_{L^2}\|[\partial^{\gamma},n^0]\nabla\phi^1\|_{L^2}\\
\leq & C(\|\phi^1\|^2_{H^{3}}+\|\phi^1\|^2_{H^{s}}) +C(1+\varepsilon^2\|{\bf u}^1\|_{L^{\infty}}^2) (\varepsilon^2\|\partial^{\gamma}\Delta\phi^1\|_{L^2}^2)\\
\leq & C(1+\varepsilon^2\|{\bf u}^1\|_{H^2}^2)|||\phi^1|||^2_{\varepsilon,s}+C\|\phi^1\|^2_{H^{3}},
\end{split}
\end{equation*}
thanks to \eqref{assumption1}, the Sobolev embedding $H^2\hookrightarrow L^{\infty}$ and the definition of the triple norm \eqref{trinorm}.

For the first term $I_{212}$, by integration by parts and Lemma \ref{lem2}, we have
\begin{equation*}
\begin{split}
|I^{\varepsilon}_{212}|=&|\frac{\varepsilon^2}{2} \int\nabla\cdot\left(\frac{{\bf u}^0+\varepsilon{\bf u}^1}{n^0+\varepsilon n^1}\right)|\partial^{\gamma}\Delta\phi^1|^2|\\
\leq & C(1+\varepsilon^2\|(n^1,{\bf u}^1)\|^2_{H^{3}}) \cdot\varepsilon^2\|\partial^{\gamma}\Delta\phi^1\|^2_{L^2}\\
\leq & C(1+\varepsilon^2|||({\bf u}^1,\phi^1)|||^2_{\varepsilon,3}) |||\phi^1|||^2_{\varepsilon,s}.
\end{split}
\end{equation*}

For the third term $I^{\varepsilon}_{213}$, by H\"older inequality and Sobolev embedding, we have
\begin{equation*}
\begin{split}
|I^{\varepsilon}_{213}|\leq & \|\frac{\varepsilon\partial^{\gamma}\Delta\phi^1({\bf u}^0+\varepsilon{\bf u}^1)}{n^0+\varepsilon n^1}\|_{L^2} \|\partial^{\gamma}\nabla\Delta\phi^0\|_{L^2}\\
\leq & \|\partial^{\gamma}\nabla\Delta\phi^0\|_{L^2}^2 +C(1+\varepsilon^2\|{\bf u}^1\|^2_{L^{\infty}})(\varepsilon^2\|\partial^{\gamma}\Delta\phi^1\|^2_{L^2})\\
\leq & C+ C(1+\varepsilon^2\|{\bf u}^1\|^2_{H^2})|||\phi^1|||^2_{\varepsilon,s}.
\end{split}
\end{equation*}

For the fourth term $I^{\varepsilon}_{214}$, we have by H\"older inequality,
\begin{equation*}
\begin{split}
I^{\varepsilon}_{214}\leq & C\|\sqrt{\varepsilon}\partial^{\gamma}\nabla R^1\|_{L^2}^2 +C(1+\varepsilon^2\|{\bf u}^1\|^2_{L^{\infty}}) (\varepsilon^2\|\partial^{\gamma}\Delta\phi^1\|^2_{H^s})\\
\leq & C\varepsilon(\|\nabla\phi^1\|_{H^s}^2+\|\phi^1\|^2) +C(1+\varepsilon^2\|{\bf u}^1\|^2_{H^{2}}) (\varepsilon^2\|\partial^{\gamma}\Delta\phi^1\|^2_{L^2})\\
\leq & C(1+\varepsilon^2\|{\bf u}^1\|^2_{H^2})|||\phi^1|||^2_{\varepsilon,s},
\end{split}
\end{equation*}
where we have used Lemma \ref{lem1} with $|\alpha|=s+1$ there and the fact that $\|\phi^1\|_{H^{s+1}}\approx \|\nabla\phi^1\|_{H^{s}}+\|\phi^1\|_{H^s}$. Summarizing, we obtain
\begin{equation}\label{equ33}
\begin{split}
I^{\varepsilon}_{21}\leq & C+C\|\phi^1\|^2_{H^3} +C(1+\varepsilon^2|||({\bf u}^1,\phi^1)|||^2_{\varepsilon,3}) |||\phi^1|||^2_{\varepsilon,s}.
\end{split}
\end{equation}
Putting \eqref{equ32} and \eqref{equ33} together, we end the proof of Lemma \ref{lem9}.
\end{proof}

\begin{lemma}\label{lem10}
The term $I^{\varepsilon}_1$ in \eqref{equ14} is bounded by
\begin{equation}\label{equ40}
\begin{split}
I^{\varepsilon}_{1}\leq & -\frac{\varepsilon}{2}\frac{d}{dt}\int(\frac{ n^0}{n^\varepsilon})|\partial^{\gamma}\nabla\phi^1|^2 -\frac{\varepsilon^2}{2}\frac{d}{dt}\int\frac{1}{n^\varepsilon} |\partial^{\gamma}\Delta\phi^1|^2\\
&+C(1+\varepsilon^2|||({\bf u}^1,\phi^1)|||_{\varepsilon,4}^2)(1+|||({\bf u}^1,\phi^1)|||_{\varepsilon,s}^2+|||({\bf u}^1,\phi^1)|||_{\varepsilon,3}^2),
\end{split}
\end{equation}
for all $0<\varepsilon<\varepsilon_1$, for some $\varepsilon_1>0$.
\end{lemma}
\begin{proof}
Recall that $I^{\varepsilon}_1$ is given in \eqref{equ14}. From the remainder equation \eqref{rem-p} (see also \eqref{equ18}), $I^{\varepsilon}_1$ can be divided into the following
\begin{equation}\label{equ35}
\begin{split}
I^{\varepsilon}_{1}=&-\varepsilon\int\frac{\varepsilon\partial^{\gamma}\Delta\phi^1}{n^0+\varepsilon n^1}\partial^{\gamma}\partial_t\Delta\phi^1 +\int\frac{\varepsilon\partial^{\gamma}\Delta\phi^1}{n^0+\varepsilon n^1}\partial^{\gamma}\partial_t(n^0\phi^1) \\
&+\int\frac{\varepsilon\partial^{\gamma}\Delta\phi^1}{n^0+\varepsilon n^1}\partial^{\gamma}\partial_t\Delta\phi^0 +\sqrt{\varepsilon}\int\frac{\varepsilon\partial^{\gamma}\Delta\phi^1}{n^0+\varepsilon n^1}\partial^{\gamma}\partial_tR^1=:\sum_{i=1}^4I^{\varepsilon}_{1i}.
\end{split}
\end{equation}
In the following, we will estimates the RHS terms one by one.

\begin{itemize}
  \item \emph{Estimate of $I^{\varepsilon}_{11}$.}
\end{itemize}
For the first term $I^{\varepsilon}_{11}$, we have
\begin{equation*}
\begin{split}
I^{\varepsilon}_{11}=&-\frac{\varepsilon^2}{2}\frac{d}{dt}\int\frac{1}{n^0+\varepsilon n^1}|\partial^{\gamma}\Delta\phi^1|^2 +\frac{\varepsilon^2}{2}\int\partial_t(\frac{1}{n^0+\varepsilon n^1})|\partial^{\gamma}\Delta\phi^1|^2.
\end{split}
\end{equation*}
Using Lemma \ref{lem3} and Sobolev embedding, we have
\begin{equation*}
\begin{split}
\|\partial_t(\frac{1}{n^0+\varepsilon n^1})\|_{L^{\infty}}\leq & C(1+\varepsilon\|\partial_tn^1\|_{L^{\infty}})\leq C(1+\varepsilon^2|||({\bf u}^1,\phi^1)|||_{\varepsilon,3}^2),
\end{split}
\end{equation*}
which yields that
\begin{equation}\label{equ36}
\begin{split}
I^{\varepsilon}_{11}\leq &-\frac{\varepsilon^2}{2}\frac{d}{dt}\int\frac{1}{n^0+\varepsilon n^1}|\partial^{\gamma}\Delta\phi^1|^2 +C(1+\varepsilon^2|||({\bf u}^1,\phi^1)|||_{\varepsilon,3}^2)|||\phi^1|||^2_{\varepsilon,s}.
\end{split}
\end{equation}

\begin{itemize}
  \item \emph{Estimate of $I^{\varepsilon}_{12}$.}
\end{itemize}
For the second term $I^{\varepsilon}_{12}$, we have
\begin{equation*}
\begin{split}
I^{\varepsilon}_{12}=& \int\frac{\varepsilon\partial^{\gamma}\Delta\phi^1}{n^0+\varepsilon n^1}\partial^{\gamma}(\partial_tn^0\phi^1) +\int\frac{\varepsilon\partial^{\gamma}\Delta\phi^1}{n^0+\varepsilon n^1}\partial^{\gamma}(n^0\partial_t\phi^1)\\
=&\int\frac{\varepsilon\partial^{\gamma}\Delta\phi^1}{n^0+\varepsilon n^1}\partial^{\gamma}(\partial_tn^0\phi^1) +\int\frac{\varepsilon\partial^{\gamma}\Delta\phi^1}{n^0+\varepsilon n^1}n^0\partial_t\partial^{\gamma}\phi^1 +\int\frac{\varepsilon\partial^{\gamma}\Delta\phi^1}{n^0+\varepsilon n^1}[\partial^{\gamma},n^0]\partial_t\phi^1\\
=&:\sum_{i=1}^3I^{\varepsilon}_{12i}
\end{split}
\end{equation*}
For the first term $I^{\varepsilon}_{121}$, by the multiplicative estimates \eqref{commutator}, we have
\begin{equation*}
\begin{split}
\|\partial^{\gamma}(\partial_tn^0\phi^1)\|_{L^2}\leq & C(\|\partial_{t}n^0\|_{L^{\infty}}\|\phi^1\|_{H^s} +\|\partial_{t}n^0\|_{H^s}\|\phi^1\|_{L^{\infty}})\\
\leq & C(\|\phi^1\|_{H^3}+\|\phi^1\|_{H^s}),
\end{split}
\end{equation*}
which yields
\begin{equation*}
\begin{split}
I^{\varepsilon}_{121}\leq & C(\|\phi^1\|_{H^3}+\|\phi^1\|_{H^s})\|\varepsilon\Delta\phi^1\|_{H^s}\\
\leq & C(\|\phi^1\|_{H^3}^2+|||\phi^1|||_{\varepsilon,s}^2).
\end{split}
\end{equation*}
For the second term $I^{\varepsilon}_{122}$, we have by integration by parts,
\begin{equation*}
\begin{split}
I^{\varepsilon}_{122}=&-\int\nabla(\frac{ n^0}{n^0+\varepsilon n^1})\varepsilon\partial^{\gamma}\nabla\phi^1 \partial_t\partial^{\gamma}\phi^1 -\int(\frac{ n^0}{n^0+\varepsilon n^1})\varepsilon\partial^{\gamma}\nabla\phi^1 \partial_t\nabla\partial^{\gamma}\phi^1\\
=&\int\nabla(\frac{ n^0}{n^\varepsilon}) \varepsilon\partial^{\gamma+\gamma_1}\nabla\phi^1 \partial_t\partial^{\gamma-\gamma_1}\phi^1 +\int\nabla\partial^{\gamma_1}(\frac{ n^0}{n^\varepsilon})\varepsilon\partial^{\gamma}\nabla\phi^1 \partial_t\partial^{\gamma-\gamma_1}\phi^1\\
&-\frac{\varepsilon}{2}\frac{d}{dt}\int\frac{ n^0}{n^\varepsilon} |\partial^{\gamma}\nabla\phi^1|^2 +\frac{\varepsilon}{2}\int\partial_t(\frac{ n^0}{n^\varepsilon}) |\partial^{\gamma}\nabla\phi^1|^2=:I^{\varepsilon}_{122i}
\end{split}
\end{equation*}
where $\gamma_1\leq \gamma$ is a multiindex with $|\gamma_1|=1$. By Lemma \ref{lem4} and Lemma \ref{lem3}, we have
\begin{equation*}
\begin{split}
\|\partial_t\partial^{\gamma-\gamma_1}\phi^1\|^2_{L^2}\leq & C(\|\partial_t\partial^{\gamma-\gamma_1}n^1\|^2+\|\phi^1\|^2_{H^{s-1}})\\
\leq & C(1+|||({\bf u}^1,\phi^1)|||_{\varepsilon,s}^2),
\end{split}
\end{equation*}
where $\gamma_1\leq \gamma,\ |\gamma_1|=1$ and hence $|\gamma-\gamma_1|=s-1$. On the other hand, by direct computation, we have
\begin{equation*}
\begin{split}
\|\nabla(\frac{ n^0}{n^\varepsilon})\|_{L^{\infty}}\leq & C(1+\varepsilon\|\nabla n^1\|_{L^{\infty}})\leq C(1+\varepsilon^2|||\phi^1|||_{\varepsilon,3}^2),\\
\|\nabla\partial^{\gamma_1}(\frac{n^0}{n^\varepsilon}) \|_{L^{\infty}}\leq & C(1+\varepsilon^2\|\nabla n^1\|^2_{L^{\infty}} +\varepsilon\|\nabla\partial^{\gamma_1}\phi^1\|_{L^{\infty}})\\
\leq & C(1+\varepsilon^2|||\phi^1|||_{\varepsilon,4}^2),
\end{split}
\end{equation*}
where we have used Lemma \ref{lem2} and Sobolev embedding. Hence, by H\"older inequality, we obtain
\begin{equation*}
\begin{split}
I^{\varepsilon}_{1221},I^{\varepsilon}_{1222}\leq & C\|\nabla(\frac{ n^0}{n^\varepsilon})\|_{L^{\infty}} \|\varepsilon\partial^{\gamma+\gamma_1}\nabla\phi^1\|_{L^2} \|\partial_t\partial^{\gamma-\gamma_1}\phi^1\|_{L^2}\\
&+C\|\nabla\partial^{\gamma_1}(\frac{ n^0}{n^\varepsilon}) \|_{L^{\infty}}\|\varepsilon\partial^{\gamma}\nabla\phi^1\|_{L^2} \|\partial_t\partial^{\gamma-\gamma_1}\phi^1\|_{L^2}\\
\leq & C(1+\varepsilon^2|||\phi^1|||_{\varepsilon,4}^2) (1+|||({\bf u}^1,\phi^1)|||_{\varepsilon,s}^2),
\end{split}
\end{equation*}
where we have used the boundedness of the Riesz operator. Similarly,
\begin{equation*}
\begin{split}
\|\partial_t(\frac{ n^0}{n^\varepsilon})\|_{L^{\infty}}\leq & C(1+\varepsilon^2\|\partial_tn^1\|^2_{L^{\infty}})\leq C(1+\varepsilon^2|||({\bf u}^1,\phi^1)|||_{\varepsilon,3}^2),
\end{split}
\end{equation*}
which yields that
\begin{equation*}
\begin{split}
I^{\varepsilon}_{1224}\leq & C(1+\varepsilon^2|||({\bf u}^1,\phi^1)|||_{\varepsilon,3}^2)|||\phi^1|||_{\varepsilon,s}^2.
\end{split}
\end{equation*}
Therefore, $I^{\varepsilon}_{122}$ is bounded by
\begin{equation*}
\begin{split}
I^{\varepsilon}_{122}\leq -\frac{\varepsilon}{2}\frac{d}{dt} \int\frac{n^0}{n^\varepsilon}|\partial^{\gamma}\nabla\phi^1|^2 +C(1+\varepsilon^2|||({\bf u}^1,\phi^1)|||_{\varepsilon,4}^2)(1+|||({\bf u}^1,\phi^1)|||_{\varepsilon,s}^2).
\end{split}
\end{equation*}
For the third term $I^{\varepsilon}_{123}$, we have
\begin{equation*}
\begin{split}
I^{\varepsilon}_{123}\leq & C\varepsilon^2\|\Delta\phi^1\|_{H^s}^2 +C\|[\partial^{\gamma},n^0]\partial_t\phi^1\|_{L^2}^2\\
\leq & C\varepsilon^2\|\Delta\phi^1\|_{H^s}^2 +C(\|\nabla n^0\|^2_{L^{\infty}}\|\partial_t\phi^1\|_{H^{s-1}}^2 +\|\partial_t\phi^1\|_{L^{\infty}}^2\|n^0\|^2_{H^s})\\
\leq & C(1+|||({\bf u}^1,\phi^1)|||_{\varepsilon,s}^2+|||({\bf u}^1,\phi^1)|||_{\varepsilon,3}^2),
\end{split}
\end{equation*}
thanks to Lemma \ref{Le-inequ} in the second inequality and Lemma \ref{lem4} and Lemma \ref{lem3} in the last inequality. In summary, $I^{\varepsilon}_{12}$ can be bounded by
\begin{equation}\label{equ37}
\begin{split}
I^{\varepsilon}_{12}\leq &-\frac{\varepsilon}{2}\frac{d}{dt}\int\frac{ n^0}{n^\varepsilon}|\partial^{\gamma}\nabla\phi^1|^2\\
&+C(1+\varepsilon^2|||({\bf u}^1,\phi^1)|||_{\varepsilon,4}^2)(1+|||({\bf u}^1,\phi^1)|||_{\varepsilon,s}^2+|||({\bf u}^1,\phi^1)|||_{\varepsilon,3}^2).
\end{split}
\end{equation}

\begin{itemize}
  \item \emph{Estimate of $I^{\varepsilon}_{13}$.}
\end{itemize}
For the third term $I^{\varepsilon}_{13}$, it is straightforward that
\begin{equation}\label{equ38}
\begin{split}
I^{\varepsilon}_{13}\leq C(1+\varepsilon^2\|\Delta\phi^1\|_{H^s}^2).
\end{split}
\end{equation}

\begin{itemize}
  \item \emph{Estimate of $I^{\varepsilon}_{14}$.}
\end{itemize}
For the fourth term $I^{\varepsilon}_{14}$, we have
\begin{equation}\label{equ39}
\begin{split}
|I^{\varepsilon}_{14}|\leq & C\varepsilon^2\|\Delta\phi^1\|_{H^s}^2 +C\varepsilon\|\partial_tR^1\|_{H^s}^2\\
\leq & C\varepsilon^2\|\Delta\phi^1\|_{H^s}^2 +C_1(\|\phi^1\|^2_{H^s}+\varepsilon\|\partial_t\phi^1\|^2_{H^{s}})\\
\leq & C\varepsilon^2\|\Delta\phi^1\|_{H^s}^2 +C_1(\|\phi^1\|^2_{H^s}+\|\partial_tn^1\|_{H^{s-1}}^2)\\
\leq & C(1+|||({\bf u}^1,\phi^1)|||^2_{\varepsilon,s}),
\end{split}
\end{equation}
where we have used H\"older inequality in the first inequality, Lemma \ref{lem1} in the second inequality, Lemma \ref{lem4} in the third inequality and Lemma \ref{lem3} in the last inequality. Here, we also have used the fact that $\|\partial_t\phi^1\|_{H^s}\approx\|\partial_t\phi^1\|_{H^{s-1}} +\|\partial_t\partial^{\alpha}\nabla\phi^1\|_{L^2}$ with $|\gamma|=s-1$ and $\|\phi^1\|_{H^{s-1}}\leq \|\phi^1\|_{H^{s}}$ for all integers $s>0$.

By \eqref{equ35}, using \eqref{equ36}, \eqref{equ37}, \eqref{equ38} and \eqref{equ39}, we obtain the estimate \eqref{equ40} for $I^{\varepsilon}$.
\end{proof}

Now, we can end the proof of Lemma \ref{lem8}.

\begin{proof}[End of proof of Lemma \ref{lem8}]
By using \eqref{equ12}, the estimates of \eqref{equ41}, \eqref{equ42}, \eqref{equ43} and \eqref{equ44}, and Lemma \ref{lem9} and \ref{lem10}, we close the proof of Lemma \ref{lem8}.
\end{proof}

Summarizing these lemmas, we obtain the following
\begin{proposition}\label{prop2}
Let $s\geq0$ be a positive integer, $(n^1,{\bf u}^1,\phi^1)$ be a smooth solution for the system \eqref{rem}. There exists $\varepsilon_1>0$ and $C,C'>0$ such that for any $0<\varepsilon<\varepsilon_1$, there holds
\begin{equation}\label{equ46}
\begin{split}
\varepsilon\|\nabla&{\bf u}^1(t)\|^2_{H^s}+\varepsilon\|\nabla\phi^1(t)\|^2_{H^s} +\varepsilon^2\|\Delta\phi^1(t)\|^2_{H^s}\\
\leq & C'(\varepsilon\|\nabla{\bf u}^1(0)\|^2_{H^s} +\varepsilon\|\nabla\phi^1(0)\|^2_{H^s} +\varepsilon^2\|\Delta\phi^1(0)\|^2_{H^s})\\
&+C\int_0^t(1+\varepsilon^2|||({\bf u}^1,\phi^1)|||^2_{\varepsilon,4})(1+|||({\bf u}^1,\phi^1)|||^2_{\varepsilon,s}+|||({\bf u}^1,\phi^1)|||^2_{\varepsilon,3})d\tau.
\end{split}
\end{equation}
where $C'$ depends only on $\sigma'$ and $\sigma''$.
\end{proposition}
\begin{proof}
This is shown by integrating \eqref{equ45} over $[0,t]$ and summing them up for $|\gamma|\leq s$, and then using $\sigma'<n^0<\sigma''$ and ${\sigma'}/{2}<n^{\varepsilon}<2\sigma''$ for any $t\in [0,T_{\varepsilon}]$ in \eqref{blowup} and \eqref{assumption1} for $0<\varepsilon<\varepsilon_1$.
\end{proof}

\subsection{End of proof of Theorem \ref{thm}}
Now, we are in a good position to end the proof of Theorem \ref{thm}.
\begin{proof}
Let $s\geq 7$ be an integer. By Proposition \ref{prop1} and Proposition \ref{prop2} and recalling the definition of the norm \eqref{trinorm}, we obtain the following Gronwall type inequality
\begin{equation}\label{equ47}
\begin{split}
|||({\bf u}^1,\phi^1)(t)|||_{\varepsilon,s}^2\leq & C'C_{\varepsilon}(0)\\
&+C\int_0^t(1+\varepsilon^3|||({\bf u}^1,\phi^1)|||^3_{\varepsilon,5})(1+|||({\bf u}^1,\phi^1)|||^2_{\varepsilon,s})d\tau,
\end{split}
\end{equation}
where $C_{\varepsilon}(0)=|||({\bf u}^1,\phi^1)(0)|||_{\varepsilon,s}^2$ and we have used the fact that $|||\cdot|||_{\varepsilon,3}\leq |||\cdot|||_{\varepsilon,s}$ for $s\geq 3$. From \eqref{assumption}, there exists $\varepsilon_1>0$ such that for any $0<\varepsilon<\varepsilon_1$, $\varepsilon^3|||({\bf u}^1,\phi^1)|||^3_{\varepsilon,5}\leq 1$, and hence \eqref{equ47} yields
\begin{equation}\label{equ48}
\begin{split}
|||({\bf u}^1,\phi^1)(t)|||_{\varepsilon,s}^2\leq C_2C_{\varepsilon}(0) +C_2\int_0^t(1+|||({\bf u}^1,\phi^1)|||^2_{\varepsilon,s})d\tau,
\end{split}
\end{equation}
where $C_2=\max\{C',2C\}$. On the other hand, from Lemma \ref{lem2}, there exists constant $C_3>1$ such that
\begin{equation}\label{equ49}
\begin{split}
\|n^1(t)\|_{H^s}^2\leq C_3(1+|||\phi^1|||^2_{\varepsilon,s}).
\end{split}
\end{equation}

Let $C_0=\sup_{0<\varepsilon<1}C_{\varepsilon}(0)$. We let $\tilde C$ in \eqref{assumption} satisfy $\tilde C\geq 2C_3(1+C_2C_0)e^{C_2T_{\varepsilon}}$, then from \eqref{equ48}
\begin{equation*}
|||({\bf u}^1,\phi^1)|||_{\varepsilon,s}^2\leq (1+C_2C_0)e^{C_2T_{\varepsilon}}\leq \tilde C,
\end{equation*}
and hence from \eqref{equ49}
\begin{equation*}
\begin{split}
\|n^1(t)\|_{H^s}^2\leq C_3(1+(1+C_2C_0)e^{C_2T_{\varepsilon}})\leq \tilde C.
\end{split}
\end{equation*}
Then by the continuity principle, it is standard to get the uniform in $\varepsilon$ estimates for $|||(n^1,{\bf u}^1,\phi^1)|||_{\varepsilon,s}$. In particular, for every $T'<T$, $\varepsilon^{-1}(n^{\varepsilon}-n^0)$ and $\varepsilon^{-1}({\bf u}^{\varepsilon}-{\bf u}^0)$ are bounded in $L^{\infty}([0,T'];H^{s})$ and $L^{\infty}([0,T'];{\bf H}^{s})$, respectively, uniformly in $\varepsilon$ for $\varepsilon$ small enough for some $s<s'$.
\end{proof}

\appendix
\renewcommand{\theequation}{\Alph{section}.\arabic{equation}}

\section{Commutator estimates}
\setcounter{equation}{0}
We give two important inequalities which are widely used throughout this paper \cite[Lemma X1 and Lemma X4]{KP88}.
\begin{lemma}\label{Le-inequ}
Let $\alpha$ be any multi-index with $|\alpha|=k$ and $p\in (1,\infty)$. Then there exists some constant $C>0$ such that
\begin{equation}\label{commutator}
\begin{split}
\|\partial_x^{\alpha}(fg)\|_{L^p}\leq & C\{\|f\|_{L^{p_1}}\|g\|_{\dot H^{k,p_2}}+\|f\|_{\dot H^{k,p_3}}\|g\|_{L^{p_4}}\},\\
\|[\partial_x^{\alpha},f]g\|_{L^p}\leq & C\{\|\nabla f\|_{L^{p_1}}\|g\|_{\dot H^{k-1,p_2}}+\|f\|_{\dot H^{k,p_3}}\|g\|_{L^{p_4}}\},
\end{split}
\end{equation}
where $f,g\in \mathcal{S}$, the Schwartz class and $p_2,p_3\in (1,+\infty)$ such that
\begin{equation*}
\begin{split}
\frac1p=\frac1{p_1}+\frac1{p_2}=\frac1{p_3}+\frac1{p_4}.
\end{split}
\end{equation*}
\end{lemma}

\end{document}